\documentclass[a4paper,11pt]{article}

\usepackage{amsmath}
\usepackage{amsfonts}
\usepackage{amssymb}         %additional math symbols (see amsguide.tex)
\usepackage{amsopn}
\usepackage{amsthm}
\usepackage{latexsym}
\usepackage{graphicx}        %graphic and .eps files
\usepackage{mathrsfs}		%calligraphic fonts
\usepackage{psfrag}
\usepackage{algorithm}
\usepackage[noend]{algorithmic}
\usepackage{color}
\usepackage{verbatim}
\usepackage{url}
\usepackage{multirow}
\usepackage{setspace}
\usepackage{booktabs,longtable,rotating}
\usepackage{afterpage}
\usepackage{sidecap}
\usepackage{marvosym}
\usepackage[caption=false]{subfig}
\usepackage{enumerate}
\usepackage[left=1.0in,top=1.2in,right=1.0in]{geometry}
\usepackage{import}
\usepackage{mathbbol}
\usepackage{relsize}
\usepackage{lipsum}
\usepackage{array}
\usepackage{qcircuit}
\usepackage{mathtools}
\usepackage{listings}
\usepackage[usenames,dvipsnames]{xcolor}
\usepackage[breakable, theorems, skins]{tcolorbox}
\usepackage{mdframed}
\usepackage{tikz}
\tcbset{enhanced}

\DeclareSymbolFontAlphabet{\mathbb}{AMSb}
\DeclareSymbolFontAlphabet{\mathbbl}{bbold}

\DeclareFontEncoding{FML}{}{}%
\DeclareFontSubstitution{FML}{fncmi}{m}{it}%
\DeclareSymbolFont{fouriernc}{FML}{fncmi}{m}{it}%
\DeclareMathAccent{\fvec}{0}{fouriernc}{"7E}

\newtheorem{theorem}{Theorem}
\newtheorem{definition}{Definition}

\newtheorem{proposition}{Proposition}
\newtheorem{example}{Example}

\newtheorem{assumption}{Assumption}
\newtheorem*{assumption*}{Assumption}
\newtheorem{remark}{Remark}

\newcommand{\R}{\ensuremath{\mathbb{R}}}
\newcommand{\C}{\ensuremath{\mathbb{C}}}

\renewcommand{\S}{\ensuremath{\mathbb{S}}}

\newcommand{\meas}{\ensuremath{\overset{\mathsmaller{M}}{=}}}
\newcommand{\vh}{\ensuremath{\fvec{h}}}
\newcommand{\vj}{\ensuremath{\fvec{\jmath}}}
\newcommand{\vk}{\ensuremath{\fvec{k}}}
\newcommand{\va}{\ensuremath{\fvec{a}}}
\renewcommand{\v}[1]{\ensuremath{\fvec{#1}}}

\DeclarePairedDelimiter\bra{\langle}{\rvert}
\DeclarePairedDelimiter\ket{\lvert}{\rangle}
\DeclarePairedDelimiterX\braket[2]{\langle}{\rangle}{#1 \delimsize\vert #2}
\DeclarePairedDelimiterX\dotp[2]{\langle}{\rangle}{#1, #2}

\DeclareRobustCommand{\optbox}[2][gray!20]{%
  \begin{tcolorbox}[   %% Adjust the following parameters at will.
      breakable,
      parbox=false,
      left=0pt,
      right=0pt,
      top=0pt,
      bottom=0pt,
      colback=#1,
      colframe=#1,
      width=\dimexpr\textwidth\relax,
      enlarge left by=0mm,
      boxsep=5pt,
      arc=0pt,outer arc=0pt,
    ]
    #2
  \end{tcolorbox}
}

\begin{document}
\thispagestyle{empty}
\begin{center}

{\LARGE An Introduction to Quantum Computing, Without the Physics}
\par \bigskip
     {\sc Giacomo Nannicini} \\
     IBM T.J.~Watson, Yorktown Heights, NY \\
\url{nannicini@us.ibm.com}
\par \bigskip
\par Last updated: \today.
\end{center}
\par \bigskip

\begin{abstract}
  This paper is a gentle but rigorous introduction to quantum
  computing intended for discrete mathematicians. Starting from a
  small set of assumptions on the behavior of quantum computing
  devices, we analyze their main characteristics, stressing the
  differences with classical computers, and finally describe two
  well-known algorithms (Simon's algorithm and Grover's algorithm)
  using the formalism developed in previous sections. This paper does
  not touch on the physics of the devices, and therefore does not
  require any notion of quantum mechanics. Numerical examples on an
  implementation of Grover's algorithm using open-source software are
  provided.
\end{abstract}

\section{Introduction}
\label{sec:intro}
Quantum computing is a relatively new area of computing that has the
potential to greatly speed up the solution of certain
problems. However, quantum computers work in a fundamentally different
way than classical computers. This introduction aims to explain the
basic principles underpinning quantum computing. It assumes the reader
is at ease with linear algebra, and with basic concepts in classical
computing such as Turing machines, and algorithm complexity.

The literature contains many textbooks on quantum computing: a
comprehensive reference is \cite{nielsen02quantum}, whereas more
modern textbooks that aim to be more accessible to non-physicists are
\cite{mermin07quantum,rieffel07quantum}. However, those books are
time-consuming reads. There are not many short introductions that are
truly accessible to non-physicists: \cite{rieffel00introduction} is
noteworthy, as it actually uses very little physics.

The approach used in this tutorial is, as far as we are aware,
different from the literature in the sense that it abstracts {\em
  entirely} away from quantum physics: we study a quantum computing
device starting from a small set of assumptions and rigorously derive
the remaining properties, focusing on the concepts that are necessary
to discuss quantum algorithms. The assumptions are verified in the
real world because of the laws of quantum mechanics, but it is not
necessary to understand why they hold: as long as we are willing to
take a small leap of faith and believe that these assumptions are
true, the rest will follow. The exposition in this tutorial is more
formal than in other surveys in the literature, but in some sense more
mathematically precise: it defines the necessary concepts in a
rigorous ways, rather than relying on examples or intuition, and
provides formal proofs. For this reason, this material is especially
suitable for students and researchers in various branches of applied
mathematics, who will be familiar with the (as much as possible)
deductive structure of this tutorial.

It is important to emphasize that the notation used in this tutorial
is often non standard: our choices are meant to facilitate understanding
for people who are just learning the basics of the field, therefore we
are mainly concerned with clarity, rather than eliminating
redundancy. However, in a short paragraph at the end of the tutorial
we highlight some of the major differences between our notation and
what is typically found in the literature.

\subsection{Overview}
The tutorial is structured as follows.
\begin{itemize}
\item In the rest of this section we discuss notation and linear algebra preliminaries.
\item In Section \ref{sec:qubits} we define the state of a quantum
  computer.
\item In Section \ref{sec:operations} we discuss the operations that
  can be applied by a quantum computer.
\item In Section \ref{sec:simon} we analyze Simon's algorithm, which
  gives an example of a fundamental principle in quantum algorithms
  known as {\em destructive interference}.
\item In Section \ref{sec:grover} we analyze Grover's algorithm,
  showcasing {\em amplitude amplification}, another fundamental
  principle in quantum algorithms.
\item Section \ref{sec:code} shows how to implement Grover's
  algorithm using Qiskit, an open-source Python library for quantum
  computation.
\item Finally, Section \ref{sec:furtherreading} contains notes for
  further reading.
\end{itemize}
The material in this tutorial is developed to support a graduate-level
module in quantum computing. In our experience with blackboard-style
delivery in the classroom, the material can be split into four modules
of 90-120 minutes each, covering Sections
\ref{sec:intro}-\ref{sec:qubits}, Section \ref{sec:operations},
Section \ref{sec:simon}, and Section \ref{sec:grover} respectively;
plus, if desired, a hands-on class on numerics using Section
\ref{sec:code}, which usually requires 90-120 minutes as well. In the
classroom, we suggest introducing Definitions
\ref{def:binaryrep}-\ref{def:bullet} only when they are used, and of
course many of the details can be skipped, adjusting the flow as
necessary: we highlight with light gray background parts of the
material that can be skipped, or briefly summarized, without
significantly impairing understanding of subsequent parts.

\subsection{Model of computation}
\label{sec:model}
The quantum computing device is, in abstract terms, similar to a
classical computing device: it has a state, and the state of the
device evolves by applying certain operations. The model of
computation that we consider is the quantum circuit model, which works
as follows:
\begin{enumerate}
\item The quantum computer has a {\em state} that is contained in a
  quantum register and is initialized in a predefined way.
\item The state evolves by applying {\em operations} specified in
  advance in the form of an algorithm.
\item At the end of the computation, some information on the state of
  the quantum register is obtained by means of a special operation,
  called a {\em measurement}.
\end{enumerate}
All terms in italics will be the subject of the assumptions mentioned
earlier, upon which our exposition will build. Note that this type of
computing device is similar to a Turing machine, except for the
presence of a tape. It is possible to assume the presence of a tape
and be more formal in defining a device that is the quantum equivalent
of a Turing machine, but there is no need to do so for the purposes of
this work; fundamental results regarding universal quantum computers
(i.e., the quantum equivalent of a universal Turing machine) are
presented in
\cite{deutsch85quantum,yao1993quantum,bernstein1997quantum}.

We will use the quantum circuit model throughout this tutorial,
including in the numerical example of Section \ref{sec:code}. This
model of computation closely matches that of certain quantum hardware
technologies that are used by some of the major players in the field
\cite{castelvecchi2017leap}, although we should note that the hardware
is affected by noise and therefore it does not provide an exact
implementation of the theoretical model. To understand the effect of
noise, we can give the following simple, but overall quite accurate,
intuitive explanation. According to the model of computation, the
state evolves by applying operations, and some information on the
state can be extracted via a measurement; due to noise, the state may
not evolve in the desired way (e.g., applying a certain operation on
the state $s_1$ should yield the state $s_2$, but we obtain a
different state $s_3$ instead), or the information extracted by a
measurement may not be what it is supposed to be (e.g., a measurement
should produce the output $0$ with probability $p_1$, but it produces
$0$ with a different probability $p_2$ instead).

Since this tutorial aims to be ``physics-free'', we will not dicuss
the specifics of existing quantum hardware that follows the circuit
model anymore. However, we should note that a different model for
quantum computing exists, and it is the so-called adiabatic model. We
do not discuss the adiabatic model for two reasons: first, the
adiabatic and the circuit model are equivalent
\cite{aharonov2008adiabatic}, therefore we are free to choose whatever
model is more convenient; second, the circuit model is more natural
for computer scientists, and is the one used in most textbooks on
quantum computing.

\subsection{Basic definitions and notation}
\label{sec:notation}
A discussion on quantum computers requires working with the decimal
and the binary representation of integers, some bit operations, and
familiarity with the properties of the tensor product. We describe
here the necessary concepts and the notation, so that the reader can
come back to this section at any time to clarify symbols.

\begin{definition}
  \label{def:tensor}
  Given two vector spaces $V$ and $W$ over a field $K$ with bases
  $e_1,\dots,e_m$ and $f_1,\dots,f_n$ respectively, the {\em tensor
    product} $V \otimes W$ is another vector space over $K$ of
  dimension $mn$. The tensor product space is equipped with a bilinear
  operation $\otimes : V \times W \to V \otimes W$. The vector space
  $V \otimes W$ has basis $e_i \otimes f_j \; \forall i=1,\dots,m,
  j=1,\dots,n$.
\end{definition}
If the origin vector spaces are complex Euclidean spaces of the
form $\C^n$, and we choose the standard basis (consisting of the
orthonormal vectors that have a $1$ in a single position and 0
elsewhere) in the origin vector spaces, then the tensor product is
none other than the Kronecker product, which is itself a
generalization of the outer product. This is formalized next.
\begin{definition}
  Given $A \in \C^{m \times n}, B \in \C^{p \times q}$, the {\em
    Kronecker product} $A \otimes B$ is the matrix $D \in \C^{mp \times
    nq}$ defined as:
  \begin{equation*}
    D := A \otimes B = \begin{pmatrix} a_{11} B & \dots & a_{1n} B\\
      a_{21} B & \dots & a_{2n} B \\
      \vdots & & \vdots \\
      a_{m1} B & \dots & a_{mn} B
      \end{pmatrix}.
  \end{equation*}
  If we choose the standard basis over the vector spaces $\C^{m \times
    n}$ and $\C^{p \times q}$, then the bilinear operation $\otimes$ of
  the tensor product $\C^{m \times n} \otimes \C^{p \times q}$ is
  simply the Kronecker product.
\end{definition}
In this tutorial we always work with complex Euclidean spaces of the
form $\C^n$, using the standard basis. With a slight but common abuse
of notation, we will therefore use tensor product to refer to the
Kronecker and outer products.

\begin{example}
  \label{ex:tensorproduct}
  We provide an example of the tensor product for normalized vectors,
  which will link this concept to probability distributions and will
  hopefully provide a better understanding of some of the future
  material. Consider two independent discrete random variables $X$ and
  $Y$ that describe the probability of extracting numbers from two
  urns. The first urn contains the numbers $0$ and $1$, the second urn
  contains the numbers $00, 01, 10, 11$. Assume that the extraction
  mechanism is biased and therefore the outcomes do not have equal
  probability. The outcome probabilities are given below, and for
  convenience we define two vectors containing them:
  \begin{align*}
    x = \begin{pmatrix} \Pr(X = 0) \\ \Pr(X = 1) \end{pmatrix} =
    \begin{pmatrix} 0.25 \\ 0.75 \end{pmatrix} \qquad
    y = \begin{pmatrix} \Pr(Y = 00) \\ \Pr(Y = 01) \\ \Pr(Y = 10) \\ \Pr(Y = 11) \end{pmatrix} = \begin{pmatrix} 0.2 \\ 0.2 \\ 0.2 \\ 0.4 \end{pmatrix}.
  \end{align*}
  Notice that because each vector contains probabilities for all
  possibile respective outcomes, the vectors are normalized so that
  their entries sum up to 1. Then, the joint probabilities for
  simultaneously extracting numbers from the two urns are given by the
  tensor product $x \otimes y$:
  \begin{align*}
    x \otimes y = 
    \begin{pmatrix} 0.25 \\ 0.75 \end{pmatrix} \otimes \begin{pmatrix} 0.2 \\ 0.2 \\ 0.2 \\ 0.4 \end{pmatrix}  =
    \begin{pmatrix} 0.05 \\ 0.05 \\ 0.05 \\ 0.1 \\ 0.15 \\ 0.15 \\ 0.15 \\ 0.3 \end{pmatrix} =
    \begin{pmatrix} \Pr(X = 0) \Pr(Y = 00) \\ \Pr(X = 0) \Pr(Y = 01) \\ \Pr(X = 0) \Pr(Y = 10) \\ \Pr(X = 0) \Pr(Y = 11) \\ \Pr(X = 1) \Pr(Y = 00) \\ \Pr(X = 1) \Pr(Y = 01) \\ \Pr(X = 1) \Pr(Y = 10) \\ \Pr(X = 1) \Pr(Y = 11) \end{pmatrix} =
    \begin{pmatrix} \Pr(X = 0, Y = 00) \\ \Pr(X = 0, Y = 01) \\ \Pr(X = 0, Y = 10) \\ \Pr(X = 0, Y = 11) \\ \Pr(X = 1, Y = 00) \\ \Pr(X = 1, Y = 01) \\ \Pr(X = 1, Y = 10) \\ \Pr(X = 1, Y = 11) \end{pmatrix},
  \end{align*}
  where the last equality is due to the fact that $X$ and $Y$ are
  independent. The vector $x \otimes y$ is also normalized, which is
  easy to verify algebraically.
\end{example}

The next proposition states some properties of the tensor product that
will be useful in the rest of this tutorial.
\begin{proposition}
  \label{prop:tensor}
  Let $A, B : \C^{m \times m}, C, D \in \C^{n \times n}$ be linear
  transformations on $V$ and $W$ respectively, $u, v \in \C^m, w, x \in
  \C^n$, and $a, b \in \C$. The tensor product satisfies the following
  properties:
  \begin{enumerate}[(i)]
  \item $(A \otimes C)(B \otimes D) = AB \otimes CD$.
  \item $(A \otimes C)(u \otimes w) = Au \otimes Cw$.
  \item $(u + v)\otimes w = u\otimes w + v\otimes w$.
  \item $u\otimes (w + x) = u\otimes w + u\otimes x$.
  \item $(au) \otimes (bw) = ab (u \otimes w)$.
  \item $(A \otimes C)^* = A^* \otimes C^*$.
  \end{enumerate}
\end{proposition}
Above and in the following, the notation $A^*$ denotes the conjugate
transpose of $A$, which is the matrix defined as follows: $A^* :=
\bar{A}^{\top}$. Given a matrix $A$, the notation $A^{\otimes n}$
indicates the tensor product of $A$ with itself $n$ times, and the
same notation will be used for vector spaces $\S$: $$A^{\otimes n} :=
\underbrace{A \otimes A \dots \otimes A}_{n \text{ times}}, \qquad
\S^{\otimes n} := \underbrace{\S \otimes \S \dots \otimes \S}_{n
  \text{ times}}.$$

The quantum computing literature refers to a Hilbert space, typically
denoted ${\cal H}$, rather than a complex Euclidean space
$\C^n$. However, the material discussed in this tutorial does not
require any property of Hilbert spaces that is not already present in
complex Euclidean spaces, hence we stick to the more familiar concept.

We will work extensively with binary strings, using
the following definitions.
\begin{definition}
  \label{def:binaryrep}
  For any integer $q > 0$, we denote by $\vj \in \{0,1\}^q$ a binary
  string on $q$ digits, where we use the arrow to emphasize that $\vj$
  is a string of binary digits rather than an integer. Given $\vj \in
  \{0,1\}^q$, we denote its $k$-th digit by $\vj_k$.
\end{definition}
We use the notation $\v{0}$ to denote the all-zero binary string, and
$\v{1}$ to denote the all-one binary string; the size of these strings
will always be clear from the context. We use a little-endian
convention for binary strings, i.e., the first digit is the most
significant one. Thus, the binary string $\vj \in \{0,1\}^q$
corresponds to the decimal number $\sum_{k=1}^q \vj_k 2^{q-k}$.

In the rest of this tutorial, as is frequent in the quantum computing
literature, we use $\vj \in \{0,1\}^q$ to index the elements of
$2^q$-dimensional vectors; such an index is well defined because
$\{0,1\}^q$ has $2^q$ elements.
\begin{definition}
  \label{def:xor}
  For any integer $q > 0$ and binary strings $\vj, \vk \in \{0,1\}^q$, we
  denote by $\vj \oplus \vk$ the bitwise modulo 2 addition of $q$-digit
  strings (bitwise XOR), defined as:
  \begin{equation*}
    \vj \oplus \vk = \vh, \text{ with } \vh \in \{0,1\}^q \text{ and } 
    \v{h}_p = \begin{cases} 0 & \text{if } \vj_p = \vec{k}_p \\
      1 & \text{otherwise} \end{cases} \text{ for all } p=1,\dots,q.
  \end{equation*}
\end{definition}
\begin{definition}
  \label{def:bullet}
  For any integer $q > 0$ and binary strings $\vj, \vk \in \{0,1\}^q$,
  we denote by $\vj \bullet \vk$ the bitwise dot product of $q$-digit
  strings, defined as:
  \begin{equation*}
    \vj \bullet \vk = \sum_{h = 1}^{q} \vj_h \vk_h.
  \end{equation*}
\end{definition}

The last piece of notation that we need is the {\em bra-ket} notation,
used in quantum mechanics. As mentioned earlier, this tutorial will not
touch on any quantum-mechanical concepts, however there is an
undeniable advantage in the quantum notation in that it puts the most
important information in the center of the symbols, rather than
relegate it to a marginal role in the subscript or
superscript. Furthermore, a goal of this work is to equip the reader
with the necessary tools to understand quantum computing papers, hence
it is important to familiarize with the bra-ket notation.
\begin{definition}
  \label{def:braket}
  Given a complex Euclidean space $\S \equiv \C^n$, $\ket{\psi} \in \S$
  denotes a column vector, and $\bra{\psi} \in \S^*$ denotes a row
  vector that is the conjugate transpose of $\ket{\psi}$, i.e.,
  $\bra{\psi} = \ket{\psi}^*$. The vector $\ket{\psi}$ is also called
  a {\em ket}, and the vector $\bra{\psi}$ is also called a {\em bra}.
\end{definition}
Thus, an expression such as $\braket{\psi}{\phi}$ is an inner
product. The complex Euclidean spaces used in this work will be of the
form $(\C^2)^{\otimes q}$, where $q$ is a given integer. It is
therefore convenient to specify the basis elements of such spaces.
\begin{definition}
  \label{def:standardbasis}
  The standard basis for $\C^2$ is denoted by $\ket{0}_1
  = \begin{pmatrix} 1 \\ 0 \end{pmatrix}, \ket{1}_1 = \begin{pmatrix}
    0 \\ 1 \end{pmatrix}$. The standard basis for $(\C^2)^{\otimes
    q}$, which has $2^q$ elements, is denoted by $\ket{\vj}, \vj \in
  \{0,1\}^q$.
\end{definition}

According to our notation, for any $q$-digit binary string $\vj \in
\{0,1\}^q$, $\ket{\vj}$ is the $2^q$-dimensional basis vector in
$(\C^2)^{\otimes q}$ corresponding to the binary string $\vj$. Since
we always use the standard basis and the most natural order for its
vectors, it is easy to verify that for $\vj \in \{0,1\}^q$,
$\ket{\vj}$ is the basis vector with a 1 in position $\sum_{k=1}^q
\vj_k 2^{q-k} + 1$, and 0 elsewhere. For example, $\ket{101}$ is the
8-dimensional basis vector $(0 \, 0 \, 0 \, 0 \, 0 \, 1 \, 0 \,
0)^{\top}$, obtained as the tensor product $\ket{1} \otimes \ket{0}
\otimes \ket{1}$. Whenever useful for clarity, we use a subscript for
bras and kets to denote the dimension of the space that the vector
belongs to, e.g., we write $\ket{\vj}_q$ to emphasize that we are
working in a $2^q$ dimensional space (or, in other words, that the
basis elements of the space are associated with binary strings with
$q$ digits). We typically omit the subscript if the dimension of the
space is evident from the context. We provide a further example of
this notation below.
\begin{example}
  Let us write the basis elements of $(\C^2)^{\otimes 2} = \C^2 \otimes \C^2$:
  \begin{align*}
    \ket{00}_2 &= \ket{00} = \ket{0} \otimes \ket{0} = \begin{pmatrix} 1 \\ 0 \\ 0 \\ 0 \end{pmatrix} &
    \ket{01}_2 &= \ket{01} = \ket{0} \otimes \ket{1} = \begin{pmatrix} 0 \\ 1 \\ 0 \\ 0 \end{pmatrix}\\
    \ket{10}_2 &= \ket{10} = \ket{1} \otimes \ket{0} = \begin{pmatrix} 0 \\ 0 \\ 1 \\ 0 \end{pmatrix} &
    \ket{11}_2 &= \ket{11} = \ket{1} \otimes \ket{1} = \begin{pmatrix} 0 \\ 0 \\ 0 \\ 1 \end{pmatrix}.
  \end{align*}
\end{example}
In the above example we made an exception to our rule and used a
subscript to denote the dimension of the basis vectors, just to
emphasize that $\ket{00}_2$ and $\ket{00}$ are exactly the same. In
the remainder of this paper, we will write $\ket{01}$ rather
than $\ket{01}_2$ because it is clear that the basis element
$\ket{01}$ has two digits and therefore lives in the space
$(\C^2)^{\otimes 2}$.

To improve clarity when dealing with vectors in $(\C^2)^{\otimes q}$,
we always denote basis vectors using spelled-out binary strings or
Roman letters, (e.g., $\ket{01}, \ket{\vj}, \ket{\v{h}}, \ket{\v{x}},
\ket{\v{y}}$ all denote basis vectors), whereas we use Greek letters
to denote vectors that may not be basis vectors (e.g., $\ket{\psi},
\ket{\phi}$ all denote vectors that may not be basis vectors). In the
same spirit, single-digit binary numbers are always denoted with Roman
letters (e.g., $x, y, z$ denote a $0$ or a $1$).

%% There are two
%% special basis states that deserve their own symbol because they are
%% used more frequently than other ones: the basis state corresponding to
%% the all-zero binary string, denoted $\ket{0^q}$, and the basis state
%% corresponding to the all-one binary string, denote $\ket{1^q}$.

\section{Qubits and quantum states}
\label{sec:qubits}
According to our computational model, a quantum computing device has a
state that is stored in the quantum register. Qubits are the quantum
counterpart of the bits found in classical computers: a classical
computer has registers that are made up of bits, whereas a quantum
computer has a single quantum register that is made up of qubits. The
assumption that there is a single quantum register is without loss of
generality, as one can think of multiple registers as being placed
``side-by-side'' to form a single register (of course, one would then
need to specify what operations are allowed on the resulting
register). The state of the quantum register, and therefore of the
quantum computing device, is defined next.
\begin{assumption}
  \label{ass:state}
  The state of a $q$-qubit quantum register is a unit vector in
  $\left(\C^2\right)^{\otimes q} = \underbrace{\C^2 \otimes \dots
    \otimes \C^2}_{q \text{ times}}$.
\end{assumption}
\begin{remark}
  A vector $\ket{\psi} \in \C^n$ is a unit vector if $\|\ket{\psi}\| =
  \sqrt{\braket{\psi}{\psi}} = 1$.
\end{remark}
\begin{remark}
  Choosing the standard basis for $\C^2$, the state of a 1-qubit
  register ($q = 1$) can be represented as $\alpha \ket{0} + \beta
  \ket{1} = \alpha \begin{pmatrix} 1 \\ 0 \end{pmatrix} +
  \beta \begin{pmatrix} 0 \\ 1 \end{pmatrix} = \begin{pmatrix} \alpha
    \\ \beta \end{pmatrix}$ where $\alpha, \beta \in \C$ and
  $|\alpha|^2 + |\beta|^2 = 1$.
\end{remark}
\begin{remark}
  Given the standard basis for $\C^2$, a basis for
  $\left(\C^2\right)^{\otimes q}$ is given by the following $2^q$
  vectors:
  \begin{align*}
    \vert \underbrace{00 \cdots 00}_{q \text{ digits}} \rangle &=
    \underbrace{\ket{0} \otimes \dots \otimes \ket{0} \otimes
      \ket{0}}_{q \text{ times}} \\
    \vert \underbrace{00 \cdots 01}_{q \text{ digits}} \rangle &=
    \underbrace{\ket{0} \otimes \dots \otimes \ket{0} \otimes
      \ket{1}}_{q \text{ times}} \\
    & \vdots \\
    \vert \underbrace{11 \cdots 11}_{q \text{ digits}} \rangle &=
    \underbrace{\ket{1} \otimes \dots \otimes \ket{1} \otimes
      \ket{1}}_{q \text{ times}}.
  \end{align*}
  In more compact form, the vectors are denoted by $\ket{\vj}, \vj \in
  \{0,1\}^q$. The state of a $q$-qubit quantum register can then be
  represented as: $\ket{\psi} = \sum_{\vj \in \{0,1\}^q} \alpha_{\vj}
  \ket{\vj}_q$, with $\alpha_{\vj} \in \C$ and $\sum_{\vj \in
    \{0,1\}^q} |\alpha_{\vj}|^2 = 1$.
\end{remark}
For brevity, we often write ``state of $q$-qubits'' or ``$q$-qubit
state'' to refer to the state of a $q$-qubit quantum register. This is
common in the literature, where the discussion of qubits is not
necessarily limited to the context of quantum registers. By properties
of the tensor product, we will see that sometimes it is appropriate to
refer to the state of just some of the qubits of a quantum computing
device, rather than all of them, and this may still be a well-defined
concept. We will revisit this in Section~\ref{sec:entanglement}.

It is important to remark that $\left(\C^2\right)^{\otimes q}$ is a
$2^q$-dimensional space. This is in sharp contrast with the state of
classical bits: given $q$ classical bits, their state is a binary
string in $\{0,1\}^q$, which is a $q$-dimensional space. In other
words, the dimension of the state space of quantum registers grows
{\em exponentially} in the number of qubits, whereas the dimension of
the state space of classical registers grows {\em linearly} in the
number of bits. Furthermore, to represent a quantum state we need
complex coefficients: the state of a $q$-qubit quantum register is
described by $2^q$ complex coefficients, which is an enormous amount
of information compared to what is necessary to describe a $q$-bit
classical register. However, we will see in
Section~\ref{sec:measurement} that a quantum state cannot be accessed
directly, therefore even if a description of the quantum state
requires infinite precision in principle, we cannot access such
description as easily as with classical registers. In fact, as it
turns out we cannot extract more than $q$ bits of information out of a
$q$-qubit register! This will be intuitively clear after stating the
effect of quantum measurements in Section~\ref{sec:measurement}; for a
formal proof, see \cite{holevo1973bounds}.

\subsection{Basis states and superposition}
\label{sec:superposition}
We continue our study of the state of quantum registers by discussing
the concept of superposition.
\begin{definition}
  \label{def:superposition}
  We say that $q$ qubits are in a {\em basis state} if the state
  $\ket{\psi} = \sum_{\vj \in \{0,1\}^q} \alpha_{\vj} \ket{\vj}_q$ of
  the corresponding register is such that $\exists \vk :
  |\alpha_{\vk}| = 1$, $\alpha_{\vj} = 0 \; \forall \vj \neq
  \vk$. Otherwise, we say that they are in a {\em superposition}.
\end{definition}
\begin{remark}
  A simpler, more intuitive definition would be to say that a basis
  state is such that $\ket{\psi} = \ket{\vk}$ for some $\vk \in
  \{0,1\}^q$. It is acceptable to use the simpler definition if
  desired: as it turns out, even if the states $\alpha_{\vk}
  \ket{\vk}$ for some $\vk \in \{0,1\}^q$ and $|\alpha_{\vk}|^2 = 1$
  are all different in principle, they are equivalent to $\ket{\vk}$
  up to the multiplication factor $\alpha_{\vk}$ which will be seen to
  be unimportant in Example \ref{ex:globalphase}.
\end{remark}

\begin{example}
  Consider two $1$-qubit registers and their states $\ket{\psi},
  \ket{\phi}$:
  \begin{align*}
    \ket{\psi} &= \alpha_0 \ket{0} + \alpha_1 \ket{1} \\
    \ket{\phi} &= \beta_0 \ket{0} + \beta_1 \ket{1}.
  \end{align*}
  If we put these $1$-qubit registers side-by-side to form a $2$-qubit
  register, then the $2$-qubit register will be in state:
  \begin{align*}
    \ket{\psi} \otimes \ket{\phi} = \alpha_0 \beta_0 \ket{0}
    \otimes \ket{0} + \alpha_0 \beta_1 \ket{0} \otimes \ket{1} +
    \alpha_1 \beta_0 \ket{1} \otimes \ket{0} + \alpha_1 \beta_1
    \ket{1} \otimes \ket{1}.
  \end{align*}
  If both $\ket{\psi}$ and $\ket{\phi}$ are in a basis state, we have
  that either $\alpha_0$ or $\alpha_1$ is zero, and similarly either
  $\beta_0$ or $\beta_1$ is zero, while the nonzero coefficients have
  modulus one. Thus, only one of the coefficients in the expression of
  the state of $\ket{\psi} \otimes \ket{\phi}$ is nonzero, and in fact
  its modulus is one. This implies that if both $\ket{\psi}$ and
  $\ket{\phi}$ are in a basis state, $\ket{\psi} \otimes \ket{\phi}$
  is in a basis state as well. But now assume that $\alpha_0 = \beta_0
  = \alpha_1 = \beta_1 = \frac{1}{\sqrt{2}}$: the qubits $\ket{\psi}$
  and $\ket{\phi}$ are in a superposition. Then the state of
  $\ket{\psi} \otimes \ket{\phi}$ is $\frac{1}{2} \ket{00} +
  \frac{1}{2} \ket{01} + \frac{1}{2} \ket{10} + \frac{1}{2} \ket{11}$,
  which is a superposition as well. Notice that the normalization of
  the coefficients works out, as one can easily check with simple
  algebra: the tensor product of unit vectors is also a unit vector.
\end{example}

The example clearly generalizes to an arbitary number of
qubits. In fact the following proposition is trivially true:
\begin{proposition}
  For any $q$, a $q$-qubit register is in a basis state if and only if
  its state can be expressed as the tensor product of $q$ 1-qubit
  registers, each of which is in a basis state.
\end{proposition}
Notice that superposition does not have a classical equivalent: $q$
classical bits are always in a basis state, i.e., a $q$-bit
classical register will always contain exactly one of the $2^q$ binary
strings in $\{0,1\}^q$. Indeed, superposition is one of the main
features of quantum computers that differentiates them from classical
computers. The second important feature is entanglement, that will be
discussed next.

\subsection{Product states and entanglement}
\label{sec:entanglement}
We have seen that the state of a $q$-qubit register is a vector in
$\left(\C^2\right)^{\otimes q}$, which is a $2^q$ dimensional
space. Since this is a tensor product of $\C^2$, i.e., the space in
which 1-qubit states live, it is natural to ask whether moving from
single qubits to multiple qubits gained us anything at all. In other
words, we want to investigate whether the quantum states that are
representable on $q$ qubits are simply the tensor product of $q$
1-qubit states. We can answer this question by using the definitions
given above. The state of $q$ qubits is a unit vector in
$\left(\C^2\right)^{\otimes q}$, and it can be written as:
\begin{equation*}
  \ket{\psi} = \sum_{\vj \in \{0,1\}^q} \alpha_{\vj} \ket{\vj}_q, \qquad
  \sum_{\vj \in \{0,1\}^q} |\alpha_{\vj}|^2 = 1.
\end{equation*}
Now let us consider the tensor product of $q$ 1-qubit states, the
$k$-th of which is given by $\beta_{k,0} \ket{0} + \beta_{k,1}
\ket{1}$, for $k=1,\dots,q$ (the first qubit corresponds to the most
significant bit, according to the little-endian convention). Taking
the tensor product we obtain the vector:
\begin{align*}
  \ket{\phi} &= (\beta_{1,0} \ket{0} + \beta_{1,1} \ket{1}) \otimes
  (\beta_{2,0} \ket{0} + \beta_{2,1} \ket{1}) \otimes \dots \otimes
  (\beta_{q,0} \ket{0} + \beta_{q,1} \ket{1}) \\
  &= \sum_{j_1=0}^{1} \sum_{j_{2}=0}^{1} \cdots
  \sum_{j_q=0}^1 \prod_{k=1}^{q} \beta_{k,j_k} | \underbrace{j_1 j_2
    \dots j_q}_{\substack{\text{taken as a} \\ \text{binary string}}}
  \rangle = \sum_{\vj \in \{0,1\}^q} \prod_{k=1}^{q} \beta_{k,
    \vj_k} \ket{\vj}_q, \\
  & \text{satisfying } |\beta_{k,0}|^2 + |\beta_{k,1}|^2 = 1
  \; \forall k=1,\dots,q.
\end{align*}
The normalization condition for $\ket{\phi}$ implies the normalization
condition of $\ket{\psi}$, but the converse is not true. That is,
$|\beta_{k,0}|^2 + |\beta_{k,1}|^2 = 1 \; \forall k=1,\dots,q$ implies
$\sum_{j_1=0}^{1} \sum_{j_2=0}^{1} \cdots \sum_{j_q=0}^1
\left|\prod_{k=1}^{q} \beta_{k,j_k}\right|^2 = 1$, but not
viceversa. This means that there exist values of $\alpha_{\vj}$, with
$\sum_{\vj \in \{0,1\}^q} |\alpha_{\vj}|^2 = 1$, that cannot be
expressed as coefficients $\beta_{k,0}, \beta_{k,1}$ (for
$k=1,\dots,q$) satisfying the conditions for $\ket{\phi}$.

\optbox{
This is easily clarified with an example.
\begin{example}
  \label{ex:productstate}
  Consider two 1-qubit states:
  \begin{align*}
    \ket{\psi} &= \alpha_0 \ket{0} + \alpha_1 \ket{1} \\
    \ket{\phi} &= \beta_0 \ket{0} + \beta_1 \ket{1}.
  \end{align*}
  Taking the two qubits together in a 2-qubit register, the state of
  the 2-qubit register is:
  \begin{equation}
    \begin{split}
      \ket{\psi} \otimes \ket{\phi} = \alpha_0 \beta_0 \ket{00} + \alpha_0
      \beta_1 \ket{01} + \alpha_1 \beta_0 \ket{10} + \alpha_1 \beta_1
      \ket{11},
    \end{split} \label{eq:productstate}
  \end{equation}
  with the normalization conditions $|\alpha_0|^2 + |\alpha_1|^2 = 1$
  and $|\beta_0|^2 + |\beta_1|^2 = 1$.  The general state of a 2-qubit
  register $\ket{\xi}$ is:
  \begin{equation}
  \label{eq:generalstate}
  \ket{\xi} = \gamma_{00} \ket{00} + \gamma_{01} \ket{01} +
  \gamma_{10} \ket{10} + \gamma_{11} \ket{11},
  \end{equation}
  with normalization condition $|\gamma_{00}|^2 + |\gamma_{01}|^2 +
  |\gamma_{10}|^2 + |\gamma_{11}|^2 = 1$. Comparing equations
  \eqref{eq:productstate} and \eqref{eq:generalstate}, we determine that
  $\ket{\xi}$ is of the form $\ket{\psi} \otimes \ket{\phi}$ if and only if
  it satisfies the relationship:
  \begin{equation}
    \label{eq:gammarel}
    \gamma_{00} \gamma_{11} = \gamma_{01} \gamma_{10}.
  \end{equation}
  Clearly $\ket{\psi} \otimes \ket{\phi}$ yields coefficients that satisfy
  this condition. To see the converse, let $\theta_{00}, \theta_{01},
  \theta_{10}, \theta_{11}$ be the phases of $\gamma_{00}, \gamma_{01},
  \gamma_{10}, \gamma_{11}$. Notice that \eqref{eq:gammarel} implies:
  \begin{align*}
    |\gamma_{00}|^2 |\gamma_{11}|^2 &= |\gamma_{01}|^2 |\gamma_{10}|^2 \\
    \theta_{00} + \theta_{11} &= \theta_{01} + \theta_{10}.
  \end{align*}
  Then we can write:
  \begin{align*}
    |\gamma_{00}| &= \sqrt{|\gamma_{00}|^2} = 
    \sqrt{|\gamma_{00}|^2 (|\gamma_{00}|^2 + |\gamma_{01}|^2 + 
      |\gamma_{10}|^2 + |\gamma_{11}|^2)} \\
    &= \sqrt{|\gamma_{00}|^4 + |\gamma_{00}|^2 |\gamma_{01}|^2 + |\gamma_{00}|^2 |\gamma_{10}|^2 + |\gamma_{01}|^2 |\gamma_{10}|^2} \\
    &= \underbrace{\sqrt{|\gamma_{00}|^2 + |\gamma_{01}|^2}}_{|\alpha_0|}\underbrace{\sqrt{|\gamma_{00}|^2 + |\gamma_{10}|^2}}_{|\beta_0|},
  \end{align*}
  and similarly for the other coefficients:
  \begin{align*}
    |\gamma_{01}| &= \underbrace{\sqrt{|\gamma_{00}|^2 + |\gamma_{01}|^2}}_{|\alpha_0|}\underbrace{\sqrt{|\gamma_{01}|^2 + |\gamma_{11}|^2}}_{|\beta_1|} \\
    |\gamma_{10}| &= \underbrace{\sqrt{|\gamma_{10}|^2 + |\gamma_{11}|^2}}_{|\alpha_1|}\underbrace{\sqrt{|\gamma_{00}|^2 + |\gamma_{10}|^2}}_{|\beta_0|} \\
    |\gamma_{11}| &= \underbrace{\sqrt{|\gamma_{10}|^2 + |\gamma_{11}|^2}}_{|\alpha_1|}\underbrace{\sqrt{|\gamma_{01}|^2 + |\gamma_{11}|^2}}_{|\beta_1|}.
  \end{align*}
  To fully define the coefficients $\alpha_0, \alpha_1, \beta_0,
  \beta_1$ we must determine their phases. We can assign:
  \begin{align}
    \label{eq:coeffproduct}
    \begin{split}
      \alpha_0 = e^{i\theta_{00}}|\alpha_0|, \qquad \alpha_1 = e^{i\theta_{10}}|\alpha_1|, 
      \qquad
      \beta_0 = |\beta_0|, \qquad \beta_1 = e^{i(\theta_{01}-\theta_{00})}|\beta_1|.
    \end{split}
  \end{align}
  It is now easy to verify that the state $\ket{\xi}$ in
  \eqref{eq:generalstate} can be expressed as $\ket{\psi} \otimes
  \ket{\phi}$ in \eqref{eq:productstate} with coefficients $\alpha_0,
  \alpha_1, \beta_0, \beta_1$ as given in \eqref{eq:coeffproduct}.

  The condition in equation \eqref{eq:gammarel}, to verify if the
  coefficients of a $2$-qubit state $\ket{\xi}$ can be expressed as a
  tensor product of two $1$-qubit states, can also be written in
  matrix form, which makes it easier to remember. If we assign the
  rows of the matrix to the first qubit, and the columns to the second
  qubit, we can arrange the coefficients $\gamma$ as follows (notice
  how the first qubit has value 0 in the first row and 1 in the second
  row; similarly for the second qubit and the columns):
  \begin{equation*}
    \begin{pmatrix}
      \gamma_{00} & \gamma_{01} \\
      \gamma_{10} & \gamma_{11}
    \end{pmatrix}.
  \end{equation*}
  Then, $\ket{\xi}$ is a tensor product of two $1$-qubit states if and
  only if this matrix has rank 1. This is equivalent to
  \eqref{eq:gammarel}.
\end{example}
}
We formalize the concept of expressing a quantum state as a tensor
product of lower-dimensional quantum states as follows.
\begin{definition}
  \label{def:entanglement}
  A quantum state $\ket{\psi} \in \left(\C^2\right)^{\otimes q}$ is a
  {\em product state} if it can be expressed as a tensor product
  $\ket{\psi_1} \otimes \dots \otimes \ket{\psi_q}$ of $q$ 1-qubit
  states. Otherwise, it is {\em entangled}.
\end{definition}
Notice that a general quantum state $\ket{\psi}$ could be the product
of two or more lower-dimensional quantum state, e.g., $\ket{\psi} =
\ket{\psi_1} \otimes \ket{\psi_2}$, with $\ket{\psi_1}$ and
$\ket{\psi_2}$ being entangled states. In such a situation,
$\ket{\psi}$ exhibits some entanglement, but in some sense it can
still be ``simplified''. Generally, according to the definition above,
we call a quantum state entangled as long as it cannot be fully
decomposed into a tensor product of 1-qubit states. In the case of
quantum systems composed of multiple subsystems (rather than just two
subsystems as in the example $\ket{\psi} = \ket{\psi_1} \otimes
\ket{\psi_2}$), the concept of entanglement as discussed in the
literature is not as simple as given in Def.~\ref{def:entanglement}
(and the rank-1 test discussed at the end of Example
\ref{ex:productstate} is not well-defined). However, our simplified
definition works in this tutorial and for most of the literature on
quantum algorithms, therefore we can leave other considerations aside;
we refer to \cite{coffman2000distributed} as an entry point for a
discussion on multipartite entanglement.

\begin{example}
  Consider the following 2-qubit state:
  \begin{equation*}
    \frac{1}{2} \ket{00} + \frac{1}{2} \ket{01} + \frac{1}{2} \ket{10} +
    \frac{1}{2} \ket{11}.
  \end{equation*}
  This is a product state because it is equal to
  $\left(\frac{1}{\sqrt{2}} \ket{0} + \frac{1}{\sqrt{2}}
  \ket{1}\right)\otimes \left(\frac{1}{\sqrt{2}} \ket{0} +
  \frac{1}{\sqrt{2}} \ket{1}\right)$. By contrast, the following
  2-qubit state:
  \begin{equation*}
    \frac{1}{\sqrt{2}} \ket{00} + \frac{1}{\sqrt{2}} \ket{11}
  \end{equation*}
  is an entangled state, because it cannot be expressed as a product
  of two 1-qubit states.
\end{example}

\section{Operations on qubits}
\label{sec:operations}
Operations on quantum states must satisfy certain conditions, to
ensure that applying an operation does not break the basic properties
of the quantum state. The required property is stated below, and we
treat it as an assumption.
\begin{assumption}
  \label{ass:operations}
  An {\em operation} applied by a quantum computer with $q$ qubits,
  also called a {\em gate}, is a unitary matrix in $\C^{2^q \times
    2^q}$.
\end{assumption}
\begin{remark}
  A matrix $U$ is unitary if $U^* U = U U^* = I$.
\end{remark}
A well-known property of unitary matrices is that they are
norm-preserving; that is, given a unitary matrix $U$ and a vector $v$,
$\|U v\| = \|v\|$. Thus, for a $q$-qubit system, the quantum state is
a unit vector $\ket{\psi} \in \C^{2^q}$, a quantum operation is a
matrix $U \in \C^{2^q \times 2^q}$, and the application of $U$ onto
the state $\ket{\psi}$ is the unit vector $U \ket{\psi} \in
\C^{2^q}$. This leads to the following remarks:
\begin{itemize}
\item Quantum operations are {\em linear}.
\item Quantum operations are {\em reversible}.
\end{itemize}
While these properties may initially seem to be extremely restrictive,
\cite{deutsch85quantum} shows that a universal quantum computer is
Turing-complete, implying that it can simulate any Turing-computable
function with an additional polynomial amount of space, given
sufficient time. Out of the two properties indicated above, the most
counterintuitive is perhaps reversibility: the classical notion of
computation does not appear to be reversible, because memory can be
erased and, in the classical Turing machine, symbols can be erased
from the tape. However, \cite{bennett73logical} shows that all
computations (including classical computations) can be made reversible
by means of extra space. The general idea to make a function
invertible is to have separate input and output registers: any output
is stored in a different location than the input, so that the input
does not have to be erased.  This is a standard trick in quantum
computing that will be discussed in Section \ref{sec:simon}, but in
order to do that, we first need to introduce some notation for quantum
circuits.

\subsection{Notation for quantum circuits}
\label{sec:circuitnotation}
A quantum circuit is represented by indicating which operations are
performed on each qubit, or group of qubits. For a quantum computer
with $q$ qubits, we represent $q$ qubit lines, where the top line
indicates qubit $1$ and the rest are given in increasing order from
the top. Operations are represented as gates; from now, the two terms
are used interchangeably. Gates take qubit lines as input, have the
same number of qubit lines as output, and apply the unitary matrix
indicated on the gate to the quantum state of those qubits. Figure
\ref{fig:basiccircuit} is a simple example.
\begin{figure}[h!]
\leavevmode
\centering
\Qcircuit @C=1em @R=.7em {
\lstick{\text{qubit } 1} & \multigate{2}{U}  & \qw \\
\lstick{\text{qubit } 2} & \ghost{U}         & \qw \\
\lstick{\text{qubit } 3} & \ghost{U}         & \qw \\
}
\caption{A simple quantum circuit.}
\label{fig:basiccircuit}
\end{figure}
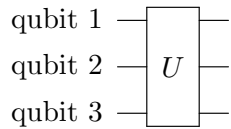

\noindent Note that circuit diagrams are read from left to right, but
because each gate corresponds to applying a matrix to the quantum
state, the matrices corresponding to the gates should be written from
right to left in the mathematical expression describing the
circuit. For example, in the circuit in Figure \ref{fig:circuitorder},
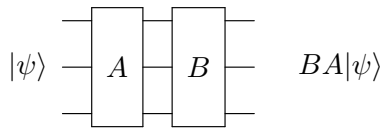
\begin{figure}[h!]
\leavevmode
\centering
\Qcircuit @C=1em @R=.7em {
                    & \multigate{2}{A}  & \multigate{2}{B}  & \qw & \\
\lstick{\ket{\psi}} & \ghost{A}         & \ghost{B}         & \qw & \rstick{BA\ket{\psi}}\\
                    & \ghost{A}         & \ghost{B}         & \qw & \\
}
\caption{Order of the operations in a quantum circuit.}
\label{fig:circuitorder}
\end{figure}
the outcome of the circuit is the state $BA \ket{\psi}$,
because we start with state $\ket{\psi}$, and we first apply the gate
with unitary matrix $A$, and then $B$.

Gates can also be applied to individual qubits. Because a
single qubit is a vector in $\C^2$, a single-qubit gate is a unitary
matrix in $\C^{2 \times 2}$. Consider the same three-qubit device, and
suppose we want to apply the gate only to the third qubit. We would
write it as in Figure \ref{fig:singlequbit}.
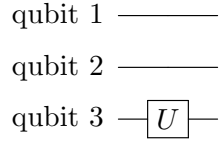
\begin{figure}[h!]
\leavevmode
\centering
\Qcircuit @C=1em @R=.5em @!R {
\lstick{\text{qubit } 1} & \qw       & \qw \\
\lstick{\text{qubit } 2} & \qw       & \qw \\
\lstick{\text{qubit } 3} & \gate{U}  & \qw \\
}
\caption{A circuit with a single-qubit gate.}
\label{fig:singlequbit}
\end{figure}

\noindent From an algebraic point of view, the action of our first
example in Figure \ref{fig:basiccircuit} on the quantum state is
clear: the state of the three qubits is mapped onto another three-qubit
state, as $U$ acts on all the qubits. To understand the example in
Figure \ref{fig:singlequbit}, where $U$ is a single-qubit gate that
acts on qubit 3 only, we must imagine that an identity gate is applied
to all the empty qubit lines. Therefore, Figure \ref{fig:singlequbit}
can be thought of as indicated in Figure \ref{fig:singlequbitI}.
\begin{figure}[h!]
\leavevmode
\centering
\Qcircuit @C=1em @R=.7em {
\lstick{\text{qubit } 1} & \gate{I} & \qw \\
\lstick{\text{qubit } 2} & \gate{I}  & \qw \\
\lstick{\text{qubit } 3} & \gate{U}  & \qw \\
}
\caption{Equivalent representation of a circuit with a single-qubit gate.}
\label{fig:singlequbitI}
\end{figure}
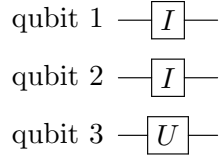

\noindent This circuit can be interpreted as applying the gate $I
\otimes I \otimes U$ to the 3-qubit state. Notice that by
convention the matrix $U$, which is applied to qubit 3, appears in the
rightmost term of the tensor product. This is because qubit 3 is
associated with the least significant digit according to our
little-endian convention, see Def.~\ref{def:binaryrep} and the subsequent discussion. If we have a
product state $\ket{\psi} \otimes \ket{\phi} \otimes \ket{\xi}$, we can write
labels as indicated in Figure \ref{fig:circuitproduct}.
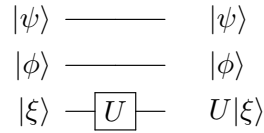
\begin{figure}[h!]
\leavevmode
\centering
\Qcircuit @C=1em @R=.3em @!R {
\lstick{\ket{\psi}} & \qw       & \qw & \rstick{\ket{\psi}} \\
\lstick{\ket{\phi}} & \qw       & \qw & \rstick{\ket{\phi}} \\
\lstick{\ket{\xi}} & \gate{U}  & \qw & \rstick{U\ket{\xi}} \\
}
\caption{Effect of a single-qubit gate on a product state.}
\label{fig:circuitproduct}
\end{figure}

\noindent Indeed, $(I \otimes I \otimes U)(\ket{\psi} \otimes
\ket{\phi} \otimes \ket{\xi}) = \ket{\psi} \otimes \ket{\phi} \otimes
U\ket{\xi}$. If the system is in an entangled state, however, the
action of $(I \otimes I \otimes U)$ cannot be determined in such a
simple way, because the state cannot be factored as a product
state. Thus, for a general entangled input state, the effect of the
circuit is as indicated in Figure \ref{fig:circuitentangled}.
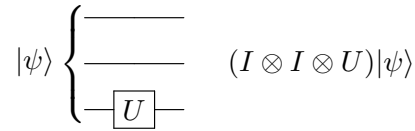
\begin{figure}[h!]
\leavevmode
\centering
\Qcircuit @C=1em @R=.3em @!R {
 & \qw  & \qw  \\
 & \qw  & \qw & \rstick{(I \otimes I \otimes U)\ket{\psi}} \\
 & \gate{U}  & \qw  
  \inputgroupv{1}{3}{.8em}{1.5em}{\ket{\psi}}
  {\gategroup{1}{3}{3}{3}{.8em}{\}}}\\
}
\caption{Effect of a single-qubit gate on an entangled state.}
\label{fig:circuitentangled}
\end{figure}
Notice that this fact is essentially the reason why simulation of
quantum computations on classical computers may take exponential
resources in the worst case: to simulate the effect of even a
single-qubit gate on the entangled state $\ket{\psi}$, we have to
explicitly compute the effect of the $2^q \times 2^q$ matrix $(I
\otimes I \otimes U)$ on the state $\ket{\psi}$. This requires
exponential space with a naive approach (if the matrices and vectors
are stored explicitly), and even with more parsimonious approaches it
may require exponential time (e.g., if we compute elements of the
state vector one at a time). As long as the quantum state is not
entangled computations can be carried out on each qubit independently,
but entanglement requires us to keep track of the full quantum state
in $2^q$-dimensional complex space, leading to large amounts of memory
-- or time -- required.

\subsection{Input-output, and measurement gates}
\label{sec:measurement}
We now discuss the input-output model for quantum computations. The
{\em input} of a quantum algorithm consists of an initial quantum
state and a quantum circuit.
\begin{remark}
  The quantum state and the quantum circuit must be described in a
  suitable compact way: for a circuit on $q$ qubits, a unitary matrix
  can be of size $2^q \times 2^q$, but for an efficient algorithm we
  require that the circuit contains polynomially many gates in $q$ and
  each gate has a compact representation. This will be discussed
  further in Section~\ref{sec:basicops}.
\end{remark}
By convention, the initial quantum state of the quantum computing
device is assumed to be $\ket{\v{0}}_q$ unless otherwise
specified. All algorithms described in this tutorial start from this
state. Of course, this does not prevent the algorithm from acting on
the state, transforming it into a more suitable one. Examples of how
this can be done will be seen in subsequent sections. The important
remark is that if there is any data that has to be fed to the
algorithm, this data is embedded in the quantum circuit given as part
of the input. We should also note that there are hybrid algorithms
involving classical and quantum computations. In such situations, the
quantum computations can generally be thought of as subroutines, but
this does not change the principle that each of these quantum
computations will be described by an initial quantum state (typically,
$\ket{\v{0}}_q$) and a quantum circuit. This summarizes the input
model. But what is the {\em output} of the quantum computer?

So far we characterized properties of quantum states and quantum
gates. Remarkably, the state of a $q$-qubit quantum register is
described by a vector of dimension $2^q$, exponentially larger than
the dimension of the vector required to describe $q$ classical
bits. However, there is a catch: in a classical computer we can simply
read the state of the bits, whereas in a quantum computer we do not
have direct, unrestricted access to the quantum state. Information on
the quantum state is only gathered through a measurement gate,
indicated in the circuit diagram in Figure \ref{fig:meas}.
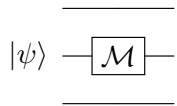
\begin{figure}[h!]
\leavevmode
\centering
\Qcircuit @C=1em @R=1em {
                    & \qw  & \qw \\
\lstick{\ket{\psi}} & \meter         & \\
                    & \qw         & \qw \\
}
\caption{Single-qubit measurement.}
\label{fig:meas}
\end{figure}
We now formally define the effect of a single-bit measurement gate.
\begin{assumption}
  \label{ass:meas}
  Information on the state of a quantum computing device can only be
  obtained through a {\em measurement}.  Given a $q$-qubit quantum
  state $\ket{\psi} = \sum_{\vj \in \{0,1\}^q} \alpha_{\vj}
  \ket{\vj}_q$, a {\em measurement gate} on qubit $k$ outputs $0$ with
  probability $\sum_{\vj \in \{0,1\}^q: \vj_k = 0} |\alpha_{\vj}|^2$,
  and $1$ with probability $\sum_{\vj \in \{0,1\}^q: \vj_k = 1}
  |\alpha_{\vj}|^2$. Let $x \in \{0,1\}$ be the measured value. After
  the measurement, the quantum state becomes: $$\sum_{\substack{\vj
      \in \{0,1\}^q:\\ \vj_k = x}} \frac{\alpha_{\vj}}{\sqrt{\sum_{\vj
        : \vj_k = x} |\alpha_{\vj}|^2}} \ket{\vj}_q.$$ The original
  quantum state is no longer recoverable.
\end{assumption}
\begin{remark}
  The state of the quantum system after a measurement collapses to a
  linear combination of only those basis states that are consistent
  with the outcome of the measurement, i.e., basis states $\ket{\vj}$
  with $\vj_k = x$. The coefficients $\alpha_{\vj}$ for such basis
  states are normalized to yield a unit vector.
\end{remark}

The rule for single-qubit measurements leads to a very simple and
natural expression for the probability of observing a given binary
string when measuring all the qubits.
\begin{proposition}
  \label{prop:meas}
  Given a $q$-qubit quantum state $\ket{\psi} = \sum_{\vj \in \{0,1\}^q}
  \alpha_{\vj} \ket{\vj}_q$, applying a measurement gate to the $q$ qubits
  in any order yields $\vj$ with probability $|\alpha_{\vj}|^2$, for $\vj
  \in \{0,1\}^q$.
\end{proposition}
\optbox{
  \begin{proof}
  We need to show that the probability of observing $\vj$ after $q$
  single-qubit measurements is equal to $|\alpha_{\vj}|^2$. We can do this
  by induction on $q$. The case $q=1$ is trivial. We now show how to
  go from $q-1$ to $q$. In terms of notation, we will write $\Pr(\text{Q}k
  \meas x)$ to denote the probability that the measurement of qubit
  $k$ yields $x \in \{0,1\}$. If it is important to indicate the
  quantum state on which the measurement is performed, we denote it as
  $\Pr_{\ket{\psi}}(\text{Q}k \meas x)$.
  
  Suppose the qubits are measured in an abitrary order and the qubit
  in position $h$ is the first to be measured. (The order of the
  remaining measurements does not matter for the proof, because after
  the first measurement we will rely on the inductive hypothesis). The
  probability of obtaining the outcome $\vj$ is:
  \begin{align*}
    \Pr_{\ket{\psi}}\left(\text{Q}1 \meas \vj_{1}, \dots, \text{Q}q \meas
    \vj_{q} \right) = \\ \Pr_{\ket{\psi}}\left(\text{Q}1 \meas
    \vj_{1}, \dots, \text{Q}(h-1) \meas \vj_{h-1}, \text{Q}(h+1) \meas \vj_{h+1}, \dots,
    \text{Q}q \meas \vj_{q} | \text{Q}h \meas \vj_{h}\right) \Pr_{\ket{\psi}}
    \left(\text{Q}h \meas \vj_{h}\right) = \\
    \Pr_{\ket{\phi}}\left(\text{Q}1 \meas
    \vj_{1}, \dots, \text{Q}(h-1) \meas \vj_{h-1}, \text{Q}(h+1) \meas \vj_{h+1}, \dots,
    \text{Q}q \meas \vj_{q} \right) \Pr_{\ket{\psi}}
    \left(\text{Q}h \meas \vj_{h}\right),
  \end{align*}
  where $\ket{\phi}$ is the state obtained from $\ket{\psi}$ after
  measuring the qubit in position $h$ and observing
  $\vj_h$. Therefore, we have:
  \begin{equation*}
    \ket{\phi} = \sum_{\substack{\vk \in \{0,1\}^q :\\ \vk_{h} = \vj_{h}}}
    \frac{\alpha_{\vk}}{\sqrt{\sum_{\vk \in \{0,1\}^q : \vk_{h} =
          \vj_{h}} |\alpha_{\vk}|^2}} \ket{\vk}_q := \sum_{\substack{\vk \in
      \{0,1\}^q :\\ \vk_{h} = \vj_{h}}} \beta_{\vk} \ket{\vk}_q,
  \end{equation*}
  and the coefficients $\beta_{\vk}$, as given above, are only
  defined for $\vk \in \{0,1\}^q : \vk_{h} = \vj_{h}$.  By the
  definition of single-qubit measurement, we also have:
  \begin{equation*}
    \Pr_{\ket{\psi}}\left(\text{Q}h \meas \vj_{h}\right) = \sum_{\vk
      \in \{0,1\}^q: \vk_{h} = \vj_{h}} |\alpha_{\vk}|^2.
  \end{equation*}
  By the induction hypothesis:
  \begin{align*}
    \Pr_{\ket{\phi}}\left(\text{Q}1 \meas
    \vj_{1}, \dots, \text{Q}(h-1) \meas \vj_{h-1}, \text{Q}(h+1) \meas \vj_{h+1}, \dots,
    \text{Q}q \meas \vj_{q} \right) = |\beta_{\vj}|^2,
  \end{align*}
  because: $\ket{\phi}$ is the state after measuring the $h$-th qubit
  and obtaining $\vj_h$ as the outcome, therefore it only contains
  basis states $\vk$ with $\vk_h = \vj_h$; and the induction
  hypothesis imposes that to compute the probability of observing
  $\vj_{\ell}, \ell \neq h$ in the respective positions we simply need
  to look at the coefficient of the corresponding basis state.
  Remembering that $\beta_{\vk} = \alpha_{\vk}/\left(
  \sqrt{\sum_{\vk \in \{0,1\}^q : \vk_{h} = \vj_{h}}
    |\alpha_{\vk}|^2}\right)$, we finally obtain:
  \begin{align*}
    \Pr_{\ket{\psi}}\left(\text{Q}1 \meas \vj_{1}, \dots, \text{Q}q \meas
    \vj_{q} \right) =  
    \frac{|\alpha_{\vj}|^2}{\displaystyle
      \sum_{\vk \in \{0,1\}^q : \vk_{h} = \vj_{h}}
      |\alpha_{\vk}|^2} \sum_{\substack{\vk \in \{0,1\}^q : \\
        \vk_{h} = \vj_{h}}} |\alpha_{\vk}|^2 = |\alpha_{\vj}|^2.
  \end{align*}
\end{proof}
}
Proposition \ref{prop:meas} above shows that the two circuits in Figure
\ref{fig:multimeas} are equivalent.
\begin{figure}[h!]
\leavevmode
\centering
\Qcircuit @C=1em @R=0.7em {
                    & \multigate{2}{\metersymb}        & \\
\lstick{\ket{\psi}} & \ghost{\metersymb}         & \\
                    & \ghost{\metersymb}         & \\
}
\hspace{10em}
\Qcircuit @C=1em @R=0.7em {
                    & \meter         & \\
\lstick{\ket{\psi}} & \meter         & \\
                    & \meter         & \\
}
\caption{Multiple-qubit measurement.}
\label{fig:multimeas}
\end{figure}
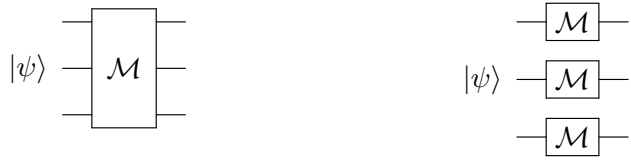

\noindent In other words, the single-qubit measurement gate is
sufficient to measure any number of qubits in the most natural way,
i.e., the measurement outcome $\vj$ on the $q$ qubits occurs with
probability that is exactly equal to the modulus squared of the state
coefficients $\alpha_{\vj}$. Notice that with this simple rule, it is
easy to compute the probability of obtaining a given string on a given
subset of the qubits: we just need to add up the modulus squared of
the coefficients for all those basis states that contain the desired
string in the desired position.

\begin{example}
  Consider again the following 2-qubit state:
  \begin{equation*}
    \alpha_{00} \ket{00} + \alpha_{01} \ket{01} + \alpha_{10} \ket{10}
    + \alpha_{11} \ket{11} = \frac{1}{2} \ket{00} + \frac{1}{2}
    \ket{01} + \frac{1}{2} \ket{10} + \frac{1}{2} \ket{11}.
  \end{equation*}
  We remarked that this is a product state. As usual, let qubit Q1 the
  first qubit (i.e., the one corresponding to the first digit in the
  two-digit binary strings), and let qubit Q2 be the second qubit
  (i.e., the one corresponding to the second digit in the two-digit
  binary strings). Then:
  \begin{align*}
    \Pr(\text{Q}1 \meas 0) &= |\alpha_{00}|^2 + |\alpha_{01}|^2 = \left(\frac{1}{2}\right)^2 + \left(\frac{1}{2}\right)^2 = \frac{1}{2} \\
    \Pr(\text{Q}1 \meas 1) &= |\alpha_{10}|^2 + |\alpha_{11}|^2 = \left(\frac{1}{2}\right)^2 + \left(\frac{1}{2}\right)^2 = \frac{1}{2} \\
    \Pr(\text{Q}2 \meas 0) &= |\alpha_{00}|^2 + |\alpha_{10}|^2 = \left(\frac{1}{2}\right)^2 + \left(\frac{1}{2}\right)^2 = \frac{1}{2} \\
    \Pr(\text{Q}2 \meas 1) &= |\alpha_{01}|^2 + |\alpha_{11}|^2 = \left(\frac{1}{2}\right)^2 + \left(\frac{1}{2}\right)^2 = \frac{1}{2}.
  \end{align*}
  Suppose we measure Q2 and we obtain 1 as the outcome of the
  measurement. Then the state of the 2-qubit system collapses to:
  \begin{equation*}
    \frac{1}{\sqrt{2}} \ket{01} + \frac{1}{\sqrt{2}} \ket{11}.
  \end{equation*}
  The outcome distribution for Q1 for this new state is:
  \begin{equation*}
    \Pr(\text{Q}1 \meas 0) = \frac{1}{2} \qquad \Pr(\text{Q}1 \meas 1) = \frac{1}{2}.
  \end{equation*}
  Hence, the probability of observing 0 or 1 when measuring qubit Q1
  did not change after the measurement.

  Consider now the following entangled 2-qubit state:
  \begin{equation*}
    \beta_{00} \ket{00} + \beta_{11} \ket{11} = \frac{1}{\sqrt{2}}
    \ket{00} + \frac{1}{\sqrt{2}} \ket{11}.
  \end{equation*}
  Doing the calculations, we still have:
  \begin{align*}
    \Pr(\text{Q}1 \meas 0) &= |\beta_{00}|^2 = \frac{1}{2} 
    & \Pr(\text{Q}1 \meas 1) &= |\beta_{11}|^2 = \frac{1}{2} \\
    \Pr(\text{Q}2 \meas 0) &= |\beta_{00}|^2 = \frac{1}{2} & 
    \Pr(\text{Q}2 \meas 1) &= |\beta_{11}|^2 = \frac{1}{2}.
  \end{align*}
  Suppose we measure qubit Q2 and we obtain 1 as the outcome of the
  measurement. Then the state of the 2-qubit system collapses to:
  \begin{equation*}
    \ket{11}.
  \end{equation*}
  If we measure Q1 from this state, we obtain:
  \begin{equation*}
    \Pr(\text{Q}1 \meas 0) = 0 \qquad \Pr(\text{Q}1 \meas 1) = 1.
  \end{equation*}
  The situation is now very different: the probability of the outcomes
  from a measurement on Q1 have changed after measuring Q2. This is
  exactly the concept of entanglement: when two or more qubits are
  entangled, they affect each other, and measuring one qubit changes
  the probability distribution for the other qubits.
\end{example}
The example above can be seen in terms of conditional probabilities:
if, for all $x, y \in \{0,1\}$, we have $\Pr(\text{Q}1 \meas x) =
\Pr(\text{Q}1 \meas x | \text{Q}2 \meas y)$, then the two qubits are
not entangled (product state), whereas if $\Pr(\text{Q}1 \meas x) \neq
\Pr(\text{Q}2 \meas x | \text{Q}2 \meas y)$ for some $x, y$, there is
entanglement. Indeed, recall from Example~\ref{ex:tensorproduct} that
taking the tensor product of two vectors containing outcome
probabilities for independent random variables yields the joint
probability distribution. Quantum state vectors do not contain outcome
probabilities, but the modulus squared of the components of the state
vector corresponds to a probability. Furthermore, for any two complex
numbers $\alpha, \beta \in \C$ we have $|\alpha\beta|^2 =
|\alpha|^2|\beta|^2$, so the operation applied to compute
probabilities from state coefficients is distributive with respect to
multiplication.  A product state is a tensor product of
smaller-dimensional state vectors, hence it leads to outcome
probabilities that are simply the product of the outcome probabilities
corresponding to measuring each of the qubits
independently. Conversely, an entangled state is not a product state,
and the outcomes of measuring each of the qubits are no longer
independent.
\begin{remark}
  Despite the above discussion, it would be wrong to think of the
  quantum state as a probability distribution: the quantum state {\em
    induces} a probability distribution by taking the modulus squared
  of its entries, but it is not a probability distribution! Indeed,
  the coefficients in a quantum state are complex numbers unrestricted
  in sign, while probabilities are nonnegative real numbers.
  Furthermore, just as there is an infinite set of complex numbers
  that have the same modulus (i.e., the set $\{a \in \C: |a| = v\}$
  for some real number $v > 0$ is infinite), there is an infinite
  number of quantum state vectors in $(\C^2)^{\otimes q}$ that yield
  the same distribution. Some of these may yield the same outcome of
  the computation, but others may not: this is discussed in the next
  two examples.
\end{remark}
\begin{example}
  \label{ex:globalphase}
  Suppose we have two $q$-qubit quantum states $\ket{\psi},
  \ket{\phi}$ satisfying $\ket{\psi} = e^{i\theta} \ket{\phi}$ for
  some $\theta \in \R$. Now consider the application of some unitary
  matrix $U$ onto $\ket{\psi}$ and $\ket{\phi}$, followed by a
  measurement of all the qubits. Define:
  \begin{equation*}
    U \ket{\phi} := \sum_{\vj \in \{0,1\}^q} \alpha_{\vj} \ket{\vj}
  \end{equation*}
  for some (normalized) coefficients $\alpha_{\vj}$, which implies:
  \begin{equation*}
    U \ket{\psi} = Ue^{i\theta} \ket{\phi} = \sum_{\vj \in \{0,1\}^q}
    e^{i\theta} \alpha_{\vj} \ket{\vj}.
  \end{equation*}
  This means that for a given $\vk$:
  \begin{equation*}
    \Pr_{\ket{\phi}} (\text{Q}1 \meas \vk_1, \dots, \text{Q}q \meas \vk_q) = |\alpha_{\vk}|^2,
    \qquad
    \Pr_{\ket{\psi}} (\text{Q}1 \meas \vk_1, \dots, \text{Q}q \meas \vk_q) = |e^{i\theta} \alpha_{\vk}|^2 = |\alpha_{\vk}|^2,
  \end{equation*}
  so the probability of obtaining $\vk$ as the outcome of a
  measurement is the same for both $\ket{\psi}$ and
  $\ket{\phi}$. Since this is true after applying an arbitrary unitary
  $U$, it is also true after applying a whole circuit, which is just a
  sequence of unitaries. Hence, if the vectors $\ket{\psi},
  \ket{\phi}$ satisfy the relationship $\ket{\psi} = e^{i\theta}
  \ket{\phi}$, they induce the same outcome distribution. The factor
  $e^{i\theta}$ is usually called {\em global phase} and can be safely
  be ignored.
\end{example}
\begin{example}
  Consider the following two 1-qubit state vectors:
  \begin{equation*}
    \ket{\psi} = \frac{1}{\sqrt{2}} \ket{0} + \frac{1}{\sqrt{2}} \ket{1} \qquad
    \ket{\phi} = \frac{1}{\sqrt{2}} \ket{0} - \frac{1}{\sqrt{2}} \ket{1}.
  \end{equation*}
  Both induce the same probability distribution on the measurement
  outcomes:
  \begin{align*}
    \Pr_{\ket{\psi}}(\text{Q}1 \meas 0) = \frac{1}{2} \qquad \Pr_{\ket{\psi}}(\text{Q}1 \meas 1) = \frac{1}{2} \\
    \Pr_{\ket{\phi}}(\text{Q}1 \meas 0) = \frac{1}{2} \qquad \Pr_{\ket{\phi}}(\text{Q}1 \meas 1) = \frac{1}{2}.
  \end{align*}
  But $\ket{\psi}$ and $\ket{\phi}$ are very different states! If we
  apply a certain unitary matrix to both (this gate is called Hadamard
  gate, as we will see in Section~\ref{sec:basicops}), we obtain very
  different results -- orthogonal vectors, in fact:
  \begin{align*}
    \frac{1}{\sqrt{2}} \begin{pmatrix} 1 & 1 \\ 1 & -1 \end{pmatrix} \ket{\psi} &= \frac{1}{2} \begin{pmatrix} 1 & 1 \\ 1 & -1 \end{pmatrix} \begin{pmatrix} 1 \\ 1 \end{pmatrix} = \frac{1}{2} \begin{pmatrix} 2 \\ 0 \end{pmatrix} = \ket{0} \\
    \frac{1}{\sqrt{2}} \begin{pmatrix} 1 & 1 \\ 1 & -1 \end{pmatrix} \ket{\phi} &= \frac{1}{2} \begin{pmatrix} 1 & 1 \\ 1 & -1 \end{pmatrix} \begin{pmatrix} 1 \\ -1 \end{pmatrix} = \frac{1}{2} \begin{pmatrix} 0 \\ 2 \end{pmatrix} = \ket{1}.
  \end{align*}
  This illustrates the danger of thinking about the quantum state as a
  probability distribution.
\end{example}

\subsection{The no-cloning principle}
Because measurement destroys the quantum state, it is natural to look
for a way to create a copy of a quantum state. If a clone could be
created, it would be possible to perform measurements on the clone, so
that the original state would not be destroyed. Furthermore, cloning
would allow us to take several measurements of the same set of qubits
without having to repeat the circuit that creates the quantum state.
However, it turns out that cloning is impossible: this is a direct
consequence of the properties of quantum gates, in particular the fact
that gates are unitary matrices.
\begin{proposition}
  Let $\ket{\psi}$ be an arbitrary quantum state on $q$ qubits. There
  does not exist a unitary matrix that maps $\ket{\psi}_q \otimes
  \ket{\vec{0}}_q$ to $\ket{\psi}_q \otimes \ket{\psi}_q$.
\end{proposition}
\optbox{
\begin{proof}
  Suppose there exists such a unitary $U$. Then for any two quantum
  states $\ket{\psi}_q, \ket{\phi}_q$, we have:
  \begin{align*}
    U(\ket{\psi}_q \otimes \ket{\vec{0}}_q) &= \ket{\psi}_q \otimes \ket{\psi}_q \\
    U(\ket{\phi}_q \otimes \ket{\vec{0}}_q) &= \ket{\phi}_q \otimes \ket{\phi}_q.
  \end{align*}
  Using these equalities, and remembering that $U^*U = I$, we can
  write:
  \begin{align*}
    \braket{\phi}{\psi} &= \braket{\phi}{\psi}\braket{\vec{0}}{\vec{0}} = 
    \braket{\phi}{\psi} \otimes \left(\braket{\vec{0}}{\vec{0}}\right) =
    (\bra{\phi}_q \otimes \bra{\vec{0}}_q)(\ket{\psi}_q \otimes \ket{\vec{0}}_q) 
    \\
    &= (\bra{\phi}_q \otimes \bra{\vec{0}}_q)U^*U(\ket{\psi}_q \otimes \ket{\vec{0}}_q)  = (\bra{\phi}_q \otimes \bra{\phi}_q)(\ket{\psi}_q \otimes \ket{\psi}_q) =
    \braket{\phi}{\psi}^2.    
  \end{align*}
  But $\braket{\phi}{\psi} = \braket{\phi}{\psi}^2$ is only true if
  $\braket{\phi}{\psi}$ is equal to 0 or to 1, contradicting the fact
  that $\ket{\phi}, \ket{\psi}$ are arbitrary quantum states.
\end{proof}
}

The above proposition shows that we cannot copy an arbitrary quantum
state. We remark that the proof does not rule out the possibility of
constructing a gate that copies a specific quantum state. In other
words, if we know what quantum state we want to copy, one could
construct a unitary matrix to do that; but it is impossible to
construct a single unitary matrix to copy all possible states.  This
establishes that we cannot ``cheat'' the destructive effect of a
measurement by simply cloning the state before the measurement. Hence,
whenever we run a circuit that produces an output quantum state, in
general we can reproduce the output quantum state only by repeating
all the steps of the algorithm.

\subsection{Basic operations and universality}
\label{sec:basicops}
Quantum computation does not allow the user to specify just any
unitary matrix in the code, just as classical computations do not
allow the user to specify any classical function. Rather, the user is
limited to gates (unitary matrices) which are efficiently specifiable
and implementable, just as classically one can only write efficient
programs by specifying a polynomial-size sequence of basic operations
on bits. The specification of a unitary matrix must be done by
combining gates out of a basic set, which can be thought of as the
instruction set of the quantum computer. We will now discuss what
these basic gates are, and how they can be combined to form other
operations.

The first operations that we discuss are the {\em Pauli gates}.
\begin{definition}
  The four Pauli gates are the following single-qubit gates:
  \begin{align*}
    I &= \begin{pmatrix} 1 & 0 \\ 0 & 1 \end{pmatrix} & X &= \begin{pmatrix} 0 & 1 \\ 1 & 0 \end{pmatrix} \\
    Y &= \begin{pmatrix} 0 & -i \\ i & 0 \end{pmatrix} & Z &= \begin{pmatrix} 1 & 0 \\ 0 & -1 \end{pmatrix}. \\
  \end{align*}
\end{definition}
\begin{proposition}
  The Pauli gates form a basis for $\C^{2 \times 2}$, they are
  Hermitian, and they satisfy the relationship $XYZ = iI$.
\end{proposition}
The proof is left as an exercise. The $X$ gate is the equivalent of a
NOT gate in classical computers, as it implements a bit (rather,
qubit) flip, changing from $\ket{0}$ to $\ket{1}$ and vice versa:
\begin{equation*}
  X \ket{0} = \ket{1} \qquad X \ket{1} = \ket{0}.
\end{equation*}
The $Z$ gate is also called a phase flip gate: it leaves $\ket{0}$
unchanged, and maps $\ket{1}$ to $-\ket{1}$.
\begin{equation*}
  Z \ket{0} = \ket{0} \qquad Z \ket{1} = -\ket{1}.
\end{equation*}

A single-qubit gate that is used in many quantum algorithms is the
so-called Hadamard gate:
\begin{equation*}
  H = \frac{1}{\sqrt{2}} \begin{pmatrix} 1 & 1 \\ 1 & -1 \end{pmatrix}.
\end{equation*}
The action of $H$ is as follows:
\begin{equation*}
  H \ket{0} = \frac{1}{\sqrt{2}}\left(\ket{0} + \ket{1}\right) \qquad
  H \ket{1} = \frac{1}{\sqrt{2}}\left(\ket{0} - \ket{1}\right)
\end{equation*}
In subsequent sections we will need an algebraic expression for the
action of Hadamard gates on basis states. The effect of $H$ on a
1-qubit basis state $\ket{x}$ (where $x = 0$ or $1$) can be
summarized as follows:
\begin{equation*}
  H \ket{x} = \frac{1}{\sqrt{2}}(\ket{0} + (-1)^x\ket{1}) =
  \frac{1}{\sqrt{2}} \sum_{k=0}^1 (-1)^{k x} \ket{k}.
\end{equation*}
This is consistent with our previous definition. If we apply
$H^{\otimes q}$ on a $q$-qubit basis state $\ket{\v{x}}_q$, we obtain:
\begin{align}
  \label{eq:hadamard}
  \begin{split}
  H^{\otimes q} \ket{\v{x}}_q &= \frac{1}{\sqrt{2^q}} \sum_{k_{1}=0}^1 \cdots \sum_{k_{q}=0}^1 (-1)^{\sum_{h=1}^{q} k_h \v{x}_{h}} \ket{k_{1}}_1 \otimes \cdots \otimes \ket{k_q}_1 \\
  &= \frac{1}{\sqrt{2^q}} \sum_{\vk \in \{0,1\}^q} (-1)^{\vk \bullet \v{x}} \ket{\vk}_q,
  \end{split}
\end{align}
where $\bullet$ is the bitwise dot product, as defined in
Section~\ref{sec:notation}. When considering multiple Hadamard gates
in parallel, we will also make use of the following relationship, that
can be easily verified using the definition:
\begin{equation}
  \label{eq:hrecursive}
  H^{\otimes q} = \frac{1}{\sqrt{2}} \begin{pmatrix} H^{\otimes q-1} & H^{\otimes q-1} \\ H^{\otimes q-1} & -H^{\otimes q-1} \end{pmatrix}.
\end{equation}
The next proposition shows one of the
reasons why the Hadamard gate is frequently employed in many quantum
algorithms.
\begin{proposition}
  Given a $q$-qubit quantum computing device initially in the state
  $\ket{\v{0}}_q$, applying the Hadamard gate to all qubits, or
  equivalently the matrix $H^{\otimes q}$, yields the uniform
  superposition of basis states $\frac{1}{\sqrt{2^q}} \sum_{\vj \in
    \{0,1\}^q} \ket{\vj}$.
\end{proposition}
\begin{proof}
  We have:
  \begin{equation*}
    H^{\otimes q} \ket{\vec{0}}_q = H^{\otimes q} \ket{0}^{\otimes q} =
    \left(H\ket{0}\right)^{\otimes q} = \left(\frac{1}{\sqrt{2}}\ket{0}
    + \frac{1}{\sqrt{2}}\ket{1}\right)^{\otimes q} =
    \frac{1}{\sqrt{2^q}} \sum_{\vj \in \{0,1\}^q} \ket{\vj}.
  \end{equation*}
\end{proof}
\begin{remark}
  The uniform superposition of the $2^q$ basis states on $q$ qubits
  can be obtained from the initial state $\ket{\vec{0}}_q$ by
  applying $q$ gates only.
\end{remark}
The multiple Hadamard can be represented by one of the equivalent
circuits given in Figure \ref{fig:hadamard}.
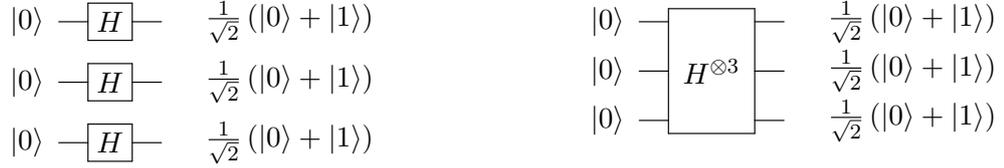
\begin{figure}[h!]
\leavevmode
\centering
\Qcircuit @C=1em @R=0.8em {
\lstick{\ket{0}} & \gate{H}  & \qw & \rstick{\frac{1}{\sqrt{2}}\left(\ket{0} + \ket{1}\right)}\\
\lstick{\ket{0}} & \gate{H}  & \qw & \rstick{\frac{1}{\sqrt{2}}\left(\ket{0} + \ket{1}\right)}\\
\lstick{\ket{0}} & \gate{H}  & \qw & \rstick{\frac{1}{\sqrt{2}}\left(\ket{0} + \ket{1}\right)}\\
}
\hspace{15em}
\Qcircuit @C=1em @R=0.8em {
\lstick{\ket{0}} & \multigate{2}{H^{\otimes 3}}  & \qw & \rstick{\frac{1}{\sqrt{2}}\left(\ket{0} + \ket{1}\right)}\\
\lstick{\ket{0}} & \ghost{H^{\otimes_3}}  & \qw & \rstick{\frac{1}{\sqrt{2}}\left(\ket{0} + \ket{1}\right)}\\
\lstick{\ket{0}} & \ghost{H^{\otimes_3}}  & \qw & \rstick{\frac{1}{\sqrt{2}}\left(\ket{0} + \ket{1}\right)}\\
}
\caption{Two representations for multiple Hadamard gates.}
\label{fig:hadamard}
\end{figure}
Many quantum algorithms (for example, the algorithms discussed in
Sections \ref{sec:simon} and \ref{sec:grover}) start by setting the
state of the quantum device to a uniform superposition, and then apply
further operations which, by linearity, are simultaneously applied to
all the possible binary strings. This is a remarkable advantage of
quantum computing over classical computing.

\optbox{

  Readers with advanced knowledge of theoretical computer science
  might be wondering how this compares to classical probabilistic
  computation: after all, probabilistic Turing machines are a
  well-known concept in computational complexity. A probabilistic
  Turing machine is initialized with a set of random bits that take an
  unknown value and influence the state transition. The state is
  described by a probability distribution over all the possible
  states, because we do not know the value of the random bits with
  which the machine is initialized. When a state transition occurs, to
  update the description of the state we need to apply the transition
  to all states that appear with positive probability. In this sense,
  operations in a probabilistic Turing machine can be thought of as
  being simultaneously applied to many (possibly all) binary
  strings. However, a probabilistic Turing machine admits a more
  compact description of the state: if we know the random bits with
  which the machine is initialized, then the state becomes
  deterministically known. Hence, for a given value of the random
  bits, the state of the probabilistic Turing machine can be described
  in linear space, and operations map one state into another state. On
  the other hand, it is not known how to obtain such a compact
  description for a quantum computer: there is no equivalent for the
  random bits, and a characterization of the state truly requires an
  exponential number of complex coefficients. In fact, it is believed
  that quantum computers are more powerful than probabilistic Turing
  machines, although there is no formal proof.

}

\optbox{

  To conclude our discussion on single-qubit gates, we remark that all
  single-qubit can be represented by the following parameterized
  matrix that describes all unitary matrices (up to a global phase
  factor, see Example~\ref{ex:globalphase}):
  \begin{equation*}
    U(\theta, \phi, \lambda) = 
    \begin{pmatrix}
      e^{-i(\phi + \lambda)/2} \cos (\theta/2) &  -e^{-i(\phi - \lambda)/2} \sin (\theta/2) \\
      e^{i(\phi - \lambda)/2} \sin (\theta/2) & e^{i(\phi + \lambda)/2} \cos (\theta/2)
    \end{pmatrix}
  \end{equation*}
  All single-qubit gates can be obtained by an appropriate choice of
  parameters $\theta, \phi, \lambda$. The open-source library Qiskit
  used for the numerical experiments in Section~\ref{sec:code} allows
  specifying a single-qubit gate in terms of $\theta, \phi, \lambda$, if
  so desired.
}

Another fundamental gate is the CNOT gate, also called ``controlled
NOT''. The CNOT gate is a two-qubit gate that has a control bit and a
target bit, and acts as follows: if the control bit is $\ket{0}$,
nothing happens, whereas if the control bit is $\ket{1}$, the target
bit is bit-flipped (i.e., the same effect
as the $X$ gate). The corresponding circuit is given in Figure
\ref{fig:cnot}.
\begin{figure}[h!]
\leavevmode
\centering
\Qcircuit @C=1em @R=0.7em {
  & \ctrl{1}  & \qw & \\
  & \targ  & \qw & 
}
\caption{The $\text{CNOT}_{12}$, or controlled NOT, gate with control
  qubit 1 and target qubit 2.}
\label{fig:cnot}
\end{figure}
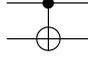

\noindent The matrix description of the gate with control qubit 1 and
target qubit 2 is as follows:
\begin{equation*}
  \text{CNOT}_{12} = 
  \begin{pmatrix}
    1 & 0 & 0 & 0 \\ 0 & 1 & 0 & 0 \\ 0 & 0 & 0 & 1 \\ 0 & 0 & 1 & 0
  \end{pmatrix}.
\end{equation*}
We can easily see that the effect of CNOT is as follows:
\begin{align*}
  \text{CNOT}_{12} \ket{00} &= \ket{00} & \text{CNOT}_{12} \ket{01} &= \ket{01} \\
  \text{CNOT}_{12} \ket{10} &= \ket{11} & \text{CNOT}_{12} \ket{11} &= \ket{10}.
\end{align*}
An interesting feature of the CNOT gate is that it can be used to swap
two qubits. A swap between two qubits Q$i$ and Q$j$ is defined as the
operation that maps a quantum state into a new quantum state in which
every basis state has its $i$-th and $j$-th digit permuted. If two
qubits are in a product state $\ket{\psi}_1 \otimes \ket{\phi}_1$,
then $\text{SWAP}(\ket{\psi}_1 \otimes \ket{\phi}_1) = \ket{\phi}_1
\otimes \ket{\psi}_1$.  Considering that CNOT, like all quantum gates,
is a linear map, it may sound surprising that it can implement a
swap. However, the SWAP gate can indeed be constructed out of CNOTs as
depicted in Figure \ref{fig:swap}.
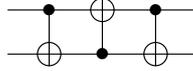
\begin{figure}[h!]
\leavevmode
\centering
\Qcircuit @C=1em @R=0.7em {
  & \ctrl{1}  & \targ      & \ctrl{1} & \qw & \\
  & \targ     & \ctrl{-1}  & \targ    & \qw & 
}
\caption{A circuit that swaps two qubits.}
\label{fig:swap}
\end{figure}
\begin{proposition}
  The circuit in Figure~\ref{fig:swap}, constructed with three CNOTs,
  swaps qubits 1 and 2.
\end{proposition}
\begin{proof}
  By linearity, it suffices to show that the circuit above maps
  $\ket{00} \to \ket{00}, \ket{01} \to \ket{10}, \ket{10} \to
  \ket{01}$, and $\ket{11} \to \ket{11}$. We have:
  \begin{align*}
    \text{CNOT}_{12} \text{CNOT}_{21} \text{CNOT}_{12} \ket{00} = \text{CNOT}_{12} \text{CNOT}_{21} \ket{00} 
    = \text{CNOT}_{12} \ket{00} = \ket{00}. \\
    \text{CNOT}_{12} \text{CNOT}_{21} \text{CNOT}_{12} \ket{01} = \text{CNOT}_{2}1 \text{CNOT}_{21} \ket{01} 
    = \text{CNOT}_{12} \ket{11} = \ket{10}. \\
    \text{CNOT}_{12} \text{CNOT}_{21} \text{CNOT}_{12} \ket{10} = \text{CNOT}_{12} \text{CNOT}_{21} \ket{11} 
    = \text{CNOT}_{12} \ket{01} = \ket{01}. \\
    \text{CNOT}_{12} \text{CNOT}_{21} \text{CNOT}_{12} \ket{11} = \text{CNOT}_{12} \text{CNOT}_{21} \ket{10} 
    = \text{CNOT}_{12} \ket{10} = \ket{11}.
  \end{align*}
  Therefore, the SWAP circuit maps:
  \begin{equation*}
  \alpha_{00} \ket{00} + \alpha_{01} \ket{01} + \alpha_{10} \ket{10} +
  \alpha_{11} \ket{11} \to  \alpha_{00} \ket{00} + \alpha_{01} \ket{10} +
  \alpha_{10} \ket{01} + \alpha_{11} \ket{11}.
  \end{equation*}
\end{proof}
The SWAP circuit is particularly important for practical reasons: in
the current generation of quantum computing hardware, two-qubit gates
can only be applied among certain pairs of qubits. For example, when
employing one of the most prevalent quantum hardware technologies
(superconducting qubits, see
e.g.~\cite{castelvecchi2017leap,devoret2013superconducting}),
two-qubit gates can only be applied to qubits that are physically
adjacent on a chip.  Thanks to the SWAP, as long as the connectivity
graph of the qubits on the device is a connected graph, two-qubit
gates can be applied to any pair of qubits: if the qubits are not
directly connected on the graph (e.g., physically located next to each
other on the chip), we just need to SWAP one of them as many times as
is necessary to bring it to a location adjacent to the other
qubit. This way, we can assume that each qubit can interact with all
other qubits from a theoretical point of view, even if from a
practical perspective this may require extra SWAP gates.

A set of gates consisting of (some) single-qubit gates plus CNOT can
be shown to be sufficient to construct any unitary matrix with
arbitrary precision. This is the concept of {\em universality}.
\begin{definition}
  A unitary matrix $A$ is an $\epsilon$-approximation of a unitary
  matrix $U$ if $\sup_{\psi : \| \psi \| = 1} \| (U - A)\psi\| <
  \epsilon$. A set of gates that can be used to construct an
  $\epsilon$-approximation of any unitary matrix, for any $\epsilon >
  0$ and on any given number of qubits, is called a {\em universal}
  set of gates.
\end{definition}
To build a universal set of gates, the first step is to show how to
construct arbitrary single-qubit gates, then use these gates to build
larger ones.
\begin{theorem} (Solovay-Kitaev \cite{kitaev97quantum,nielsen02quantum})
  \label{thm:sk}
  Let $U \in \C^{2 \times 2}$ be an arbitrary unitary matrix. Then
  there exists a constant $c$ such that there exists a sequence $S$ of
  gates of length $O(\log^c \frac{1}{\epsilon})$ that yields an
  $\epsilon$-approximation of $U$ and consists only of $H$, $T
  = \begin{pmatrix} 1 & 0 \\ 0 & e^{i\frac{\pi}{4}} \end{pmatrix}$ and
  CNOT gates.
\end{theorem}
The theorem implies that just two single-qubit gates together with
CNOT allow us to build any single-qubit gate with arbitrary
precision. \cite{dawson05solovay} gives a proof with $c \approx
3.98$. More recent work gives improved algorithm with smaller $c$, in
fact even $c = 1$ (but different constants), see
\cite{selinger12efficient,kliuchnikov16practical}. To go from
single-qubit gates to general $q$-qubit gates, one needs at most
$O(q^2 4^q)$ basic gates (i.e., the gates of Theorem \ref{thm:sk});
intuitively, this is because each gate on $q$ qubits has $2^q \times
2^q$ elements, and it takes $q^2$ basic gates to ``fill'' an arbitrary
element of a large matrix -- for a detailed discussion, see
\cite[Ch.~4]{nielsen02quantum}. In other words, the set of gates
consisting of just $H, T$ and CNOT is universal. This shows that with
a very small set of basic gates, we can construct any unitary matrix
in any dimension, although this may require many operations. Once
again, this is important for practical reasons: the current generation
of quantum hardware only allows (natively) certain single-qubit gates
and CNOT gates, but all other gates can be constructed from these.

We conclude our discussion on basic operations with a quantum circuit
for the logic AND gate. Although this gate is not used by the
algorithm described later, it is used in the numerical examples. We
already know that the $X$ gate performs the logic NOT: having access
to the AND guarantees that we can construct any Boolean circuit (since
we stated already in Section \ref{sec:operations} that quantum
computers are Turing-complete, they can of course perform Boolean
logic). The quantum version of the AND gate is the CCNOT
(doubly-controlled NOT) gate, that acts on three qubits: it has two
control qubits, and it flips (bit flip, i.e., as the $X$ gate) the
third qubit if and only if both control qubits are $\ket{1}$. The gate
is depicted in Figure \ref{fig:ccnot}.
\begin{figure}[h!]
\leavevmode
\centering
\Qcircuit @C=1em @R=0.7em @!R {
  & \ctrl{2}  & \qw & \\
  & \ctrl{1}  & \qw & \\
  & \targ  & \qw & 
}
\caption{The CCNOT, or doubly-controlled NOT, gate.}
\label{fig:ccnot}
\end{figure}
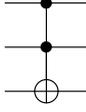
The action of CCNOT can be described as: $\ket{x}_1 \otimes \ket{y}_1
\otimes \ket{z}_1 \to \ket{x}_1 \otimes \ket{y}_1 \otimes \ket{z
  \oplus (x \cdot y)}_1$, where $x, y, z \in \{0,1\}$. Notice that if
$z = 0$, CCNOT indeed computes the logical AND between $x$ and $y$
because $0 \oplus (x \cdot y) = x \wedge y$.

\optbox{

  Of course, following our earlier discussion, CCNOT can be
  constructed using only the basic gates indicated in Theorem
  \ref{thm:sk}. For this, we can use the circuit in Figure
  \ref{fig:ccnotbasic}, see \cite{nielsen02quantum}. In this circuit
  we also use the conjugate transpose $T^*$ of the $T$ gate, but it is
  easy to see that if we really want to stick to the gates $H, T$,
  CNOT only, $T^*$ can be constructed from $T$ because
  $e^{-i\frac{\pi}{4}} = e^{i\frac{7\pi}{8}}$.  Verifying correctness
  of the construction in Figure \ref{fig:ccnotbasic} requires a few
  calculations, that we leave as an exercise. One way is to carry out
  the matrix multiplications; another way, probably more manageable if
  doing calculations by hand, is to use linearity and look at the
  effect of the circuit on each of the $2^3$ possible basis states. We
  show only part of the calculations here. Suppose the circuit is
  applied to the basis state $\ket{11x}$ with $x \in \{0,1\}$. After
  performing several simplifications ($T$ and $T^*$ cancel out, the
  $T$ gate has no effect on a qubit in state $\ket{0}$, and we can
  transform the CNOTs on the third qubit line into $X$ gates because
  we already know that the first and second qubit are in state
  $\ket{1}$), we find out that the circuit maps:
  \begin{equation*}
    \ket{1} \otimes \ket{1} \otimes \ket{x} \to (T\ket{1}) \otimes
    (T\ket{1}) \otimes (H T X T^* X T X T^* X H \ket{x}).
  \end{equation*}
  Doing the calculations, we see that:
  \begin{align*}
    H T X T^* X T X T^* X H = \begin{pmatrix} 0 & -i \\ -i & 0 \end{pmatrix},
  \end{align*}
  so that the mapping reads:
  \begin{align*}
    \ket{1} \otimes \ket{1} \otimes \ket{1} \to &(T\ket{1}) \otimes
    (T\ket{1}) \otimes (H T X T^* X T X T^* X H \ket{1}) = \\
    & (e^{i\frac{\pi}{4}} \ket{1}) \otimes (e^{i\frac{\pi}{4}} \ket{1}) \otimes (-i \ket{0}) = \ket{1} \otimes \ket{1} \otimes \ket{0} \\
    \ket{1} \otimes \ket{1} \otimes \ket{0} \to &(e^{i\frac{\pi}{4}}\ket{1})
    \otimes (e^{i\frac{\pi}{4}}\ket{1}) \otimes (-i \ket{1}) =
    \ket{1} \otimes \ket{1} \otimes \ket{1}.
  \end{align*}  
  In general, coming up with these constructions requires a good deal
  of experience, or a piece of code implementing the algorithms
  referenced above.  }
\begin{figure}[h!]
  \begin{mdframed}[linecolor=gray!20,backgroundcolor=gray!20]
  \leavevmode
  \centering
  \Qcircuit @C=1em @R=0.7em @!R {
    & \qw      & \qw      & \qw        & \ctrl{2} & \qw      & \qw      & \qw        & \ctrl{2} & \qw        & \ctrl{1} & \qw        & \ctrl{1} & \gate{T} & \qw      & \qw & \\
    & \qw      & \ctrl{1} & \qw        & \qw      & \qw      & \ctrl{1} & \qw        & \qw      & \gate{T^*} & \targ    & \gate{T^*} & \targ    & \gate{T} & \gate{T} & \qw &\\
    & \gate{H} & \targ    & \gate{T^*} & \targ    & \gate{T} & \targ    & \gate{T^*} & \targ    & \gate{T}   & \gate{H} & \qw        & \qw      & \qw      & \qw      & \qw
  }
  \caption{Decomposition of CCNOT in terms of the universal set of gates of Theorem \ref{thm:sk}.}
  \label{fig:ccnotbasic}
  \end{mdframed}
\end{figure}
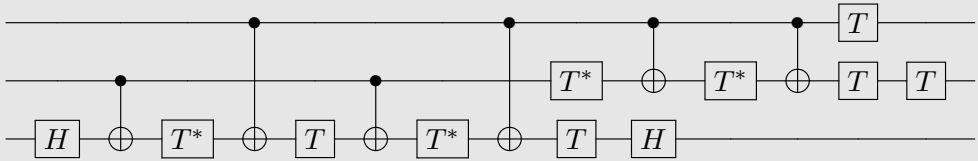

\optbox{
\subsection{Can we solve NP-hard problems?}
It is important to remark that even if we can easily create a uniform
superposition of all basis states, the rules of measurement imply that
using just this easily-obtained superposition does not allow us
satisfactorily solve NP-complete problems such as, for example, SAT
(the satisfiability problem). Indeed, suppose we have a quantum
circuit $U_f$ that encodes a SAT formula on $q$ boolean variables; in
other words, a unitary $U_f : \ket{\vj}_q \otimes \ket{0}_1 \to
\ket{\vj}_q \otimes \ket{f(\vj)}_1$, where $f(\vj)$ is 1 if the binary
string $\vj$ satisfies the formula, and 0 if not. We might be tempted
to apply $H^{\otimes q}$ to the initial state $\ket{\v{0}}_q$ to
create the uniform superposition $\frac{1}{\sqrt{2^q}} \sum_{\vj \in
  \{0,1\}^q} \ket{\vj}$, apply $U_f$ to this superposition (which
evaluates the truth assignment of all possible binary strings), and
then perform a measurement on all $q+1$ qubits. But measuring the
state:
\begin{equation*}
  U_f \left(\frac{1}{\sqrt{2^q}}
\sum_{\vj \in \{0,1\}^q} \ket{\vj} \otimes \ket{0}_1\right) = \frac{1}{\sqrt{2^q}}
\sum_{\vj \in \{0,1\}^q} \ket{\vj} \otimes \ket{f(\vj)}
\end{equation*}
will return a binary string that satisfies the formula if and only if
the last qubit has value 1 after the measurement, and this happens
with a probability that depends on the number of binary assignments
that satisfy the formula. If the SAT problem at hand is solved by
exactly $\rho$ assignments out of $2^n$ possible assignments, then the
probability of finding the solution after one measurement is
$\frac{\rho}{2^n}$: we have done nothing better than randomly sampling
a binary string and hoping that it satisfies the SAT formula. Clearly,
this is not a good algorithm. In fact, in general solving NP-hard
problems in polynomial time with quantum computers is not believed to
be possible: most researchers believe that the complexity class BQP,
which is the class of problems solvable in polynomial time by a
quantum computer with bounded (and small) error probability, does not
contain the class NP. Of course, one cannot hope to prove this
unconditionally, because showing NP $\not\subseteq$ BQP would resolve
the famous P vs NP problem. Nevertheless, it is strongly believed that
NP $\not\subseteq$ BQP, due to the lower bound on black-box search of
\cite{bennett1997strengths}, and the inability of quantum computing
researchers to develop an efficient quantum algorithm for SAT.

\par Even if we cannot solve all difficult problems in polynomial time
using a quantum computer, we will see in the next sections two
examples of quantum algorithms that are faster than any known
classical algorithm. The basic principle employed by these algorithms
is to start with a uniform superposition of basis states, then apply
operations that make the basis states interact with each other so that
the modulus of the coefficients for some (desirable) basis states
increase, which implies that the other coefficients
decrease. Performing a measurement will then reveal the solution to
the problem at hand, or some useful information about the solution,
with high probability. Of course, how to do this in order to solve a
specific problem is exactly where the crux of the matter lies.

}

\section{A simple period finding problem: Simon's algorithm}
\label{sec:simon}
In this section we describe a quantum algorithm, known as Simon's
algorithm \cite{simon97power}, that gives an expected exponential
speedup with respect to classical algorithms. Simon's algorithm is one
of the first examples of quantum speedup presented in the literature;
other early examples are
\cite{deutsch1992rapid,bernstein1997quantum}. However, we discuss
Simon's algorithm because it has many interesting features from an
educational perspective (namely, the fact that it uses both classical
and quantum computation, and that it yields an exponential speedup).

Admittedly, the problem that Simon's algorithm solves is not very
useful, but the ideas shown here give us a flavor of what quantum
computation can do. In fact, this algorithm was an inspiration for the
well-known and groundbreaking work of Shor on integer factorization
\cite{shor97polynomial}: a large part of Shor's algorithm relies on
the solution of a period finding problem, and Simon's algorithm solves
a simplified problem of the same flavor. Shor's algorithm is, however,
much more involved than Simon's algorithm, and a full treatment
requires several number-theoretical results that are beyond the scope
of this introductory material. Thus, we will focus on Simon's
algorithm, but at the end of this tutorial the reader should be
capable of undertaking the study of Shor's algorithm by himself or
herself, see also Section~\ref{sec:furtherreading}.

For Simon's algorithm, we are told that there exists a function $f :
\{0,1\}^n \to \{0,1\}^n$ with the property that $f(\v{x}) = f(\v{z})$
if and only if $\v{x} = \v{z} \oplus \v{a}$, for some unknown $\v{a}
\in \{0,1\}^n$. We do not know anything else about the function, and
the goal is to find $\v{a}$ by querying the function the smallest
number of times.
\begin{remark}
  For both Simon's algorithm and Grover's algorithm, the complexity of
  an algorithm is determined only in terms of the number of calls to
  the function $f$. Considerations on what the function $f$ actually
  implements, and how many operations are performed {\em inside} of
  $f$, or between the calls to $f$, are not part of our analysis. This
  model is known as {\em query complexity}, because -- as the name
  implies -- it defines the complexity of an algorithm as the number
  of queries to a given function (in this case, $f$). Query complexity
  is used as a model to answer important theoretical questions. There
  are many quantum algorithms that yield speedups under the query
  complexity model, but some others, e.g., Shor's algorithm, are
  faster than (known) classical algorithms under the more traditional
  computational complexity model, i.e., number of basic operations.
\end{remark}
Notice that if $\v{a} = \v{0}$ then the function is one-to-one,
whereas if $\v{a} \neq \v{0}$ the function is two-to-one, because for
every $\v{x}$, there is exactly another number in domain for which the
function has the same value. The function $f$ is assumed to be given
as a quantum circuit on $q = 2n$ qubits, via the unitary $U_f$
depicted in Figure \ref{fig:simonuf}, and we are allowed to query the
function in superposition. Remember that by linearity, to describe the
effect of $U_f$ it is enough to describe its behavior on all basis
states.
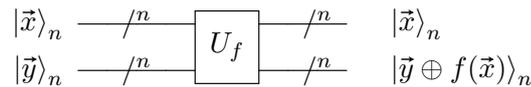
\begin{figure}[h!]
\leavevmode
\centering
\Qcircuit @C=1em @R=0.7em {
\lstick{\ket{\v{x}}_n} & \qw & {/^n} \qw & \qw & \multigate{1}{U_f} & \qw & {/^n} \qw & \qw & \rstick{\ket{\v{x}}_n} \\
\lstick{\ket{\v{y}}_n} & \qw & {/^n} \qw & \qw & \ghost{U_f}        & \qw & {/^n} \qw & \qw & \rstick{\ket{\v{y} \oplus f(\v{x})}_n}
}
\caption{The circuit implementing $U_f$ for Simon's problem, with
  basis states $\v{x}, \v{y} \in \{0,1\}^n$.}
\label{fig:simonuf}
\end{figure}

This particular form of the function, that maps $\ket{\v{x}}_n \otimes
\ket{\v{y}}_n$ to $\ket{\v{x}}_n \otimes \ket{\v{y} \oplus
  f(\v{x})}_n$, is typical of the quantum world. Notice that if
$\v{y} = \v{0}$, then $\ket{\v{y} \oplus f(\v{x})}_n =
\ket{f(\v{x})}_n$ so the circuit computes the desired
function. Furthermore, this is a reversible function, because applying
the same circuit $U_f$ goes back to the initial state:
\begin{equation*}
  U_f U_f (\ket{\v{x}}_n \otimes \ket{\v{y}}_n) = U_f (\ket{x}_n \otimes \ket{\v{y}
    \oplus f(\v{x})}_n) = \ket{\v{x}}_n \otimes \ket{\v{y} \oplus {f(\v{x})} \oplus f(\v{x})}_n =
  \ket{\v{x}}_n \otimes \ket{\v{y}}_n.
\end{equation*}

\subsection{Classical algorithm}
\label{sec:simonclassical}
Because we do not know anything about the binary string $\v{a}$, the
best we can do is to feed inputs to the function, and try to extract
information from the output. The number $\v{a}$ is determined once we
find two distinct inputs $\v{x}, \v{z}$ such that $f(\v{x}) =
f(\v{z})$, because then $\v{x} = \v{z} \oplus \v{a}$ which implies
$\v{x} \oplus \v{z} = \v{a}$.

Suppose we have evaluated $m$ distinct input values and we did not
find a match. Then $\v{a} \neq \v{x} \oplus \v{z}$ for all $\v{x},
\v{z}$ previously evaluated, therefore we have eliminated at most
$m(m-1)/2$ values of $\v{a}$. (Fewer values may have been eliminated
if we test inputs equal to $\v{x} \oplus \v{y} \oplus \v{z}$ for any
three input values $\v{x}, \v{y}, \v{z}$ already tested. In fact, if
we test $\v{w}$ such that $\v{w} = \v{x} \oplus \v{y} \oplus \v{z}$,
we have that $\v{w} \oplus \v{z} = \v{x} \oplus \v{y}$, therefore the
value $\v{w} \oplus \v{z}$ had already been eliminated from the list
of possible valus of $\v{a}$.) Since $m(m-1)/2$ is small compared to
$2^n$, the probability of success $\frac{m(m-1)}{2^{n+1}}$ is very
small until we have evaluated a number of inputs that is in the order
of $2^n$. In particular, to guarantee a probability of success of at
least $\rho$, we need $m(m-1) \ge \rho 2^{n+1}$, which implies that $m
\in O(\sqrt{\rho 2^n})$. Hence, for any positive constant $\rho$, the
number of required iterations is exponential.  After evaluating
$\frac{1 +\sqrt{2^{n+3}+1}}{2} \in O(2^{n/2})$ distinct input values
satisfying the condition outlined above for non-matching triplets (to
obtain this number, we found the smallest value of $m$ such that
$m(m-1) \ge 2^{n+1}$), we are guaranteed that a matching pair has been
found, or we can safely determine that $\v{a} = \v{0}$.

\subsection{Simon's algorithm: quantum computation}
Using a quantum computer, we can determine $\v{a}$ much faster. The idea,
first described in \cite{simon97power}, is to apply the circuit in
Figure \ref{fig:simoncircuit}.
\begin{figure}[h!]
\leavevmode
\centering
\Qcircuit @C=1em @R=0.7em {
\lstick{\ket{\v{0}}_n} & \qw & {/^n} \qw & \gate{H^{\otimes n}} & \qw & {/^n} \qw & \multigate{1}{U_f} & \qw & {/^n} \qw & \gate{H^{\otimes n}} & \qw & {/^n} \qw & \qw & \meter \\
\lstick{\ket{\v{0}}_n} & \qw & {/^n} \qw & \qw                  & \qw & \qw       & \ghost{U_f}        & \qw & {/^n} \qw & \qw & \qw & \qw & \qw & \qw \\
}
\caption{Quantum circuit used in Simon's algorithm.}
\label{fig:simoncircuit}
\end{figure}
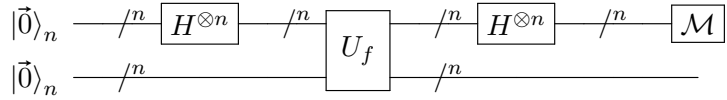

\noindent From an algebraic point of view, the circuit is described by the 
following equation:
\begin{equation*}
  (H^{\otimes n} \otimes I^{\otimes n})U_f(H^{\otimes n} \otimes
  I^{\otimes n}) (\ket{\v{0}}_n \otimes \ket{\v{0}}_n).
\end{equation*}
We now analyze the output of the quantum circuit, by looking at the
quantum states at intermediate steps of the circuit. Let $\ket{\psi}$
be the state just before the $U_f$ gate, $\ket{\phi}$ the state just
after $U_f$, and $\ket{\chi}$ the final state. In other words:
\begin{align*}
  \ket{\psi} &=   (H^{\otimes n} \otimes
  I^{\otimes n}) (\ket{\v{0}}_n \otimes \ket{\v{0}}_n) \\
  \ket{\phi} &=   U_f(H^{\otimes n} \otimes
  I^{\otimes n}) (\ket{\v{0}}_n \otimes \ket{\v{0}}_n) \\
  \ket{\chi} &=   (H^{\otimes n} \otimes I^{\otimes n})U_f(H^{\otimes n} 
  \otimes I^{\otimes n}) (\ket{\v{0}}_n \otimes \ket{\v{0}}_n).
\end{align*}
For $\ket{\psi}$, we know that $H^{\otimes n}$ creates a uniform
superposition of $\ket{\vj}_n, \vj \in \{0,1\}^n$ over the first $n$
quantum bits. Therefore we can write:
\begin{equation*}
  \ket{\psi} = (H^{\otimes n} \otimes I^{\otimes n}) (\ket{\v{0}}_n
  \otimes \ket{\v{0}}_n) = \frac{1}{\sqrt{2^n}} \sum_{\vj \in \{0,1\}^n}
  \ket{\vj}_n \otimes \ket{\v{0}}_n.
\end{equation*}
By linearity, applying $U_f$ to this state yields:
\begin{equation*}
  \ket{\phi} = U_f \ket{\psi} = \frac{1}{\sqrt{2^n}} \sum_{\vj \in
    \{0,1\}^n} \ket{\vj}_n \otimes \ket{\v{0} \oplus f(\vj)}_n =
  \frac{1}{\sqrt{2^n}} \sum_{\vj \in \{0,1\}^n} \ket{\vj}_n \otimes
    \ket{f(\vj)}_n.
\end{equation*}
We now need to analyze the effect of applying further Hadamard gates
on the top lines of the circuit. Using \eqref{eq:hadamard}, the next
step in the circuit is given by:
\begin{align}
  \ket{\chi} &= (H^{\otimes n} \otimes I^{\otimes n})
  \frac{1}{\sqrt{2^n}} \sum_{\vj \in \{0,1\}^n} \ket{\vj}_n \otimes
  \ket{f(\vj)}_n = \notag \\
  &= \frac{1}{2^n} \sum_{\vj \in \{0,1\}^n} \sum_{\vk \in \{0,1\}^n} 
  (-1)^{\vk \bullet \vj} \ket{\vk}_n \otimes
  \ket{f(\vj)}_n. \label{eq:simonbasecoeff}
\end{align}
When we make a measurement on the top $n$ qubit lines of $\ket{\chi}$
(i.e., the first $n$-qubit register, containing qubits 1 through $n$),
we obtain one of the stings $\vk$ with probability equal to the sum of
the modulus squared of the coefficient of the states $\ket{\vk}_n
\otimes \ket{f(\vj)}$, for all $\vj$. It is easier to analyze the case
$\va \neq \v{0}$ first: we will deal with the case $\va = 0$ later in
Section \ref{sec:simonanalysis}. Assuming $\va \neq \v{0}$,
$\ket{\vk}_n \otimes \ket{f(\vj)}_n = \ket{\vk}_n \otimes \ket{f(\vj
  \oplus \v{a})}_n$. Let $R$ be a set with the following property: for
every $\vj \in \{0,1\}^n$, $R$ contains either $\vj$ or $\vj \oplus
\v{a}$, but not both. (For the reader familiar with the concept of
quotient sets, $R$ is the quotient set $\{0,1\}^n/\sim$ where $\sim$
is the equivalence relationship defined as: $\v{x} \sim \v{y}$ if and
only if $\v{x} = \v{y} \oplus \v{a}$.) Then, for each $\vk$, the
string $\vk$ appears in the top qubit lines exactly in the $2^{n-1}$
basis states $\ket{\vk}_n \otimes \ket{f(\vj)}_n$ for $\vj \in R$. For
each $\vj \in R$, the coefficient of the basis state $\ket{\vk}_n
\otimes \ket{f(\vj)}_n$ is exactly the sum of the coefficients in
\eqref{eq:simonbasecoeff} for $\ket{\vk}_n \otimes \ket{f(\vj)}_n$ and
$\ket{\vk}_n \otimes \ket{f(\vj \oplus \v{a})}_n$, that is:
\begin{align*}
  \frac{(-1)^{\vk \bullet \vj} + (-1)^{\vk \bullet (\vj \oplus \v{a})}}{2^n} &= \frac{(-1)^{\vk \bullet \vj} + (-1)^{\vk \bullet \vj}(-1)^{\vk \bullet \v{a}}}{2^n} \\
  &=  \frac{(-1)^{\vk \bullet \vj}\left(1 + (-1)^{\vk \bullet \v{a}}\right)}{2^n}.
\end{align*}
Therefore the probability of obtaining the binary string $\vk$ after
measuring the top qubit lines is:
\begin{align*}
  \sum_{\vj \in R} 
  \left(\frac{(-1)^{\vk \bullet \vj}\left(1 + (-1)^{\vk \bullet \v{a}}\right)}{2^n}\right)^2 =
  2^{n-1}\left(\frac{\left(1 + (-1)^{\vk \bullet \v{a}}\right)}{2^n}\right)^2 =
  \begin{cases} \frac{1}{2^{n-1}} & \text{if } \vk \bullet \v{a} \equiv 0 \mod 2 \\
    0 & \text{if } \vk \bullet \v{a} \equiv 1 \mod 2, \end{cases}
\end{align*}
where the multiplication factor $2^{n-1}$ comes from the fact that
$|R| = \frac{2^n}{2}$. Thus, the only binary strings that have
positive probability to be observed are those strings $\vk$ for which
$\vk \bullet \v{a} \equiv 0 \mod 2$. The remaining strings are never
sampled: by carefully applying quantum operations we have reduced
their state coefficients to zero, a phenomenon known as {\em
  destructive interference}. Notice that unless $\vk = \v{0}$, then there
is a nonempty set of bits for which the modulo 2 sum of $\v{a}$ must
vanish. In this case, unless we are unlucky and we obtain the vector
$\vk = \v{0}$ (or some other undesirable cases that will be specified
later), we can express one of such bits as a modulo 2 sum of the
others, and we eliminate half of the possible values for $\v{a}$.

Our discussion shows that with a single quantum query to $U_f$, in the
case $\v{a} \neq \v{0}$ with high probability we learn very valuable
information about $\v{a}$, and we can approximately halve the search space
for $\v{a}$. It now remains to fully specify in a more precise manner how
this information can be used, and how to deal with the case $\v{a} = \v{0}$.

\subsection{Simon's algorithm: description and analysis}
\label{sec:simonanalysis}
The quantum algorithm described in the previous section yields
information on $\va$, but it does not output $\va$ directly. To recover
$\va$, further calculations have to be performed. This is a typical
situation in quantum algorithms: a quantum computation measures some
properties of the desired answer; then, classical computations are
used to analyze these properties and obtain the desired answer. Thus,
even if the quantum algorithm does not explicitly output the desired
answer, it allows us to get closer to our goal.

In the specific case of the problem discussed here, the quantum
computation allows us to learn $\vk$ such that $\vk \bullet \va \equiv
0 \mod 2$. We embed this equation into an algorithm as follows: we
initialize $E$ to the empty set; then, while the system of equations
$E$ does not have a unique solution, we apply the circuit
described in the previous section to obtain $\vk$, and add the
equation $\vk \bullet \va \equiv 0 \mod 2$ to $E$. Notice that $\va =
\v{0}$ is always a solution of the homogeneous system $E$, but we are
interested in the nonzero solutions, if any exist. In other words, we
want to determine if the null space contains any nonzero vector. We
can have two possible situations: either the system has a uniquely
determined nonzero solution $\va \neq \v{0}$, or the only possible
solution is $\v{a} = \v{0}$. Since there are $n$ unknowns and we are
dealing with a homogeneous system, to identify which of these
situations happens we need $E$ to contain $n$ linearly independent
vectors $\vk$, where independence is intended modulo 2. Because at
every iteration we obtain a random $\vk$ for which $\vk \bullet \va
\equiv 0 \mod 2$, we need to analyze how many iterations we need to
obtain $n$ such vectors with high probability.

In continuous space, uniform random sampling of vectors yields
linearly independent vectors with probability 1. In this case we are
considering linear independence among vectors that have coefficients 0
or 1, and independence is in terms of the modulo 2 sum, so the
argument is less clear; however, it is possible to show that the
probability of obtaining $n$ such linearly independent vectors after
sampling $n + t$ times is bounded below by $1 - \frac{1}{2^{t}}$
\cite[Apx.~G]{mermin07quantum}. This lower bound does not depend on
$n$. Hence, with overwhelming probability after slightly more than $n$
executions of the quantum circuit, and therefore $O(n)$ queries to the
function $f$, we determine the solution to the problem with a
classical computation that can be performed in polynomial time (i.e.,
$O(n^2)$ to determine a solution to the system of linear equations
modulo 2). We remark that once the unique nonzero $\va$ is determined,
we can easily verify that it is the solution by querying the
function. On the other hand, if $\va = \v{0}$, the algorithm will
detect that this is the case because at some point the system of
linear equations $E$ will have $\va = \v{0}$ as the only possible
solution (notice that if $\va = \v{0}$, then $\vk \bullet \va \equiv 0
\mod 2$ for any $\vk$). Compare the $O(n)$ queries of this approach
with the $O(2^{n/2})$ queries that are required by a classical
algorithm, and we have shown an exponential speedup.

This algorithm shows a typical feature of many quantum algorithms:
oftentimes, there is a classical computation to complement the quantum
computation. For example, the classical computation could verify that
the correct solution to the problem has indeed been found. In this
case, the verification is carried out by checking whether the system
of equations has a unique solution. Indeed, quantum algorithms are
probabilistic algorithm, and we can only try to increase the
probability that the correct answer is returned; only in rare cases
the solution can be obtained with probability 1, see
e.g.~\cite{brassard2002quantum}. For this reason, it is desirable to
have a way to deterministically verify correctness. This may require a
classical computation. In other words, the quantum algorithm is
applied to a problem for which it is difficult to classicaly compute
the solution, but once the solution (or some information about it) is
obtained, it is easy to classically verify that we have the right
answer. This is not known to be possible in general, since the
complexity class BQP is not known or believed to be contained in NP
(recall that NP is the class of problems admitting an efficient
classical proof that the solution has been found, i.e., a certificate
that can be checked in polynomial time); indeed, it is an active
topic of research to design verification protocols for generic quantum
computations
\cite{broadbent2009universal,aharonov2017interactive,reichardt2013classical,mahadev2018classical}. However,
the quantum algorithms presented in this tutorial will admit simple
classical verification.

\section{Black-box search: Grover's algorithm}
\label{sec:grover}
Simon's algorithm gives an exponential speedup with respect to a
classical algorithm, but it solves a very narrow problem that does not
have practical applications. We now describe an algorithm that gives
only a polynomial -- more specifically, quadratic -- speedup with
respect to classical, but it applies to a very large class of
problems. The algorithm is known as Grover's search
\cite{grover96fast}.

The problem solved by the algorithm can be described as black-box
search: we are given a circuit that computes an unknown function of a
binary string, and we want to determine for which value of the
input the function gives output 1. In other words, we are trying to
determine the unique binary string that satisfies a property encoded
by a circuit. The original paper \cite{grover96fast} describes this as
looking for a certain element in a database. Such an algorithm can be
applied whenever we are searching for a specific element in a set, we
have a way of testing if an element is the desired element (in fact,
this test must be implementable as a quantum subroutine -- see below),
and we do not have enough information to do anything smarter than a
brute force search, i.e., testing all elements in the set.

As mentioned earlier, the basic idea of the algorithm is to start with
the uniform superposition of all basis states, and iteratively
increase the coefficients of basis states that correspond to binary
strings for which the unknown function gives output 1.

We need some definitions. Let $f : \{0,1\}^n \to \{0,1\}$, and assume
that there exists a unique $\v{\ell} \in \{0,1\}^n : f(\v{\ell}) = 1$,
i.e., there is a unique element in the domain of the function that
yields output 1. We want to determine $\v{\ell}$. The function $f$ is
assumed to be encoded by a unitary as follows:
\begin{equation*}
  U_f : \ket{\vj}_n \otimes \ket{y}_1 \to \ket{\vj}_n \otimes \ket{y
    \oplus f(\vj)}_1.
\end{equation*}
As usual, we are allowed to query the function in superposition.
\begin{remark}
  Grover's search can also be applied to the case in which there are
  multiple input values that yield output 1, and we want to retrieve
  any of them; however, the analysis in that case is slightly more
  convoluted, and is not pursued here in detail. By the end of our
  analysis, the reader will have all the necessary tools to study this
  extension.
\end{remark}

\subsection{Classical algorithm}
\label{sec:groverclassical}
Given the problem definition, classical search cannot do better than
$O(2^n)$ operations. Indeed, any deterministic classical algorithm may
need to explore all $2^n$ possible input values before finding $\v{\ell}$:
given any deterministic classical algorithm, there exists a
permutation $\pi$ of $\{0,1\}^n$ that represents the longest
execution path (i.e., sequence of values at which $f$ is queried) of
such algorithm. Then, if $\v{\ell} = \pi(\v{1})$ the algorithm will
require $2^n$ queries to determine the answer, which is clearly the
worst case.

At the same time, a randomized algorithm requires $O(2^n)$ function
calls to have at least a constant positive probability to determine
$\v{\ell}$; the expected number of function calls to determine the
answer is approximately $2^{n-1}$, i.e., the expected number of draws
before we extract the black marble from an urn containing one black
marble and $2^n-1$ white marbles (sampling without replacement; the
corresponding distribution is known as the hypergeometric
distribution).

\subsection{Grover's search: algorithm description}
\label{sec:groverquantum}
The quantum search algorithm proposed in \cite{grover96fast} requires
$q = n+1$ qubits, which is exactly the number of qubits for the
encoding of $U_f$ as defined above.

The outline of the algorithm is as follows. The algorithm starts with
the uniform superposition of all basis states on $n$ qubits. The last
qubit ($n+1$) is used as an auxiliary qubit, and it is initialized to
$H\ket{1}$. We obtain the quantum state $\ket{\psi}$. Then, these
operations are repeated several times:
\begin{enumerate}[(i)]
\item Flip the sign of the vectors for which $U_f$ gives output 1.
\item Invert all the coefficients of the quantum state around the
  average coefficient -- we will explain the precise mapping
  implemented by this operation in Section~\ref{sec:inversionavg}.
\end{enumerate}
A full cycle of the two operations above increases the coefficient of
$\ket{\v{\ell}}_n \otimes \frac{1}{\sqrt{2}}(\ket{0} - \ket{1})$, and
after a certain number of cycles (to be specified later), the
coefficient of the state $\ket{\v{\ell}}_n \otimes
\frac{1}{\sqrt{2}}(\ket{0} - \ket{1})$ is large enough that it can be
obtained from a measurement with probability close to 1. This
phenomenon is known as {\em amplitude amplification}.

\begin{figure}[t!b]
  \centering
  \subfloat[Initialization.]{\includegraphics[width=0.32\textwidth]{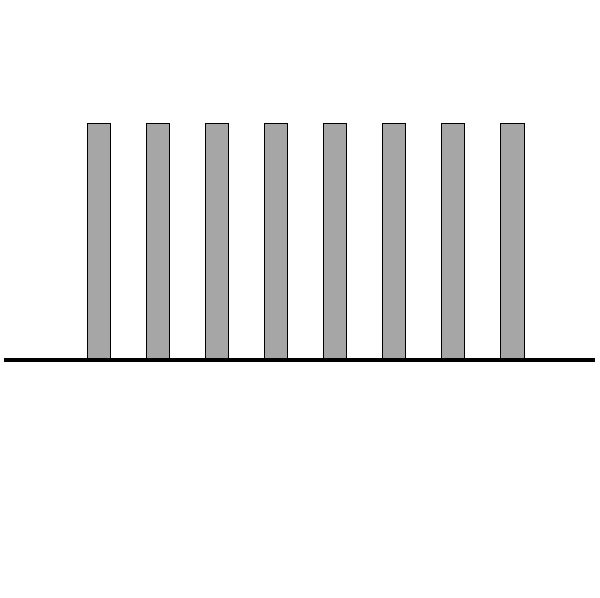}\label{fig:gr1}}\hspace{5em}
  \subfloat[Sign flip.]{\includegraphics[width=0.32\textwidth]{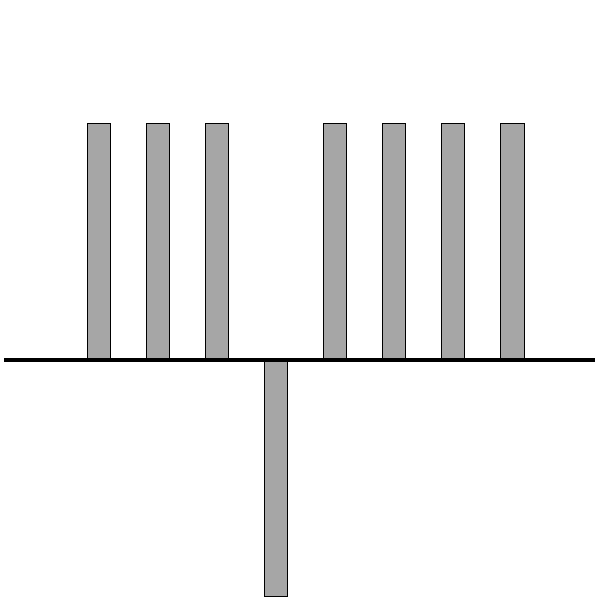}\label{fig:gr2}}\\[2em]
  \subfloat[Computation of the average.]{\includegraphics[width=0.32\textwidth]{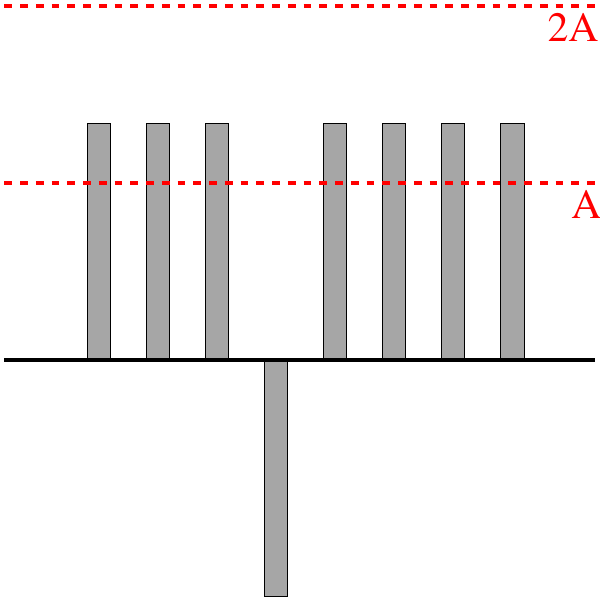}\label{fig:gr3}}\hspace{5em}
  \subfloat[Inversion about the average.]{\includegraphics[width=0.32\textwidth]{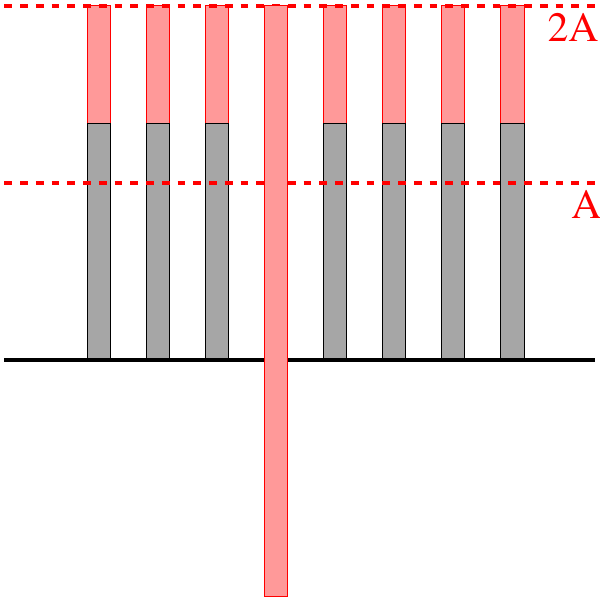}\label{fig:gr4}}
  \caption{Sketch of Grover's algorithm. The bars represent the coefficients of the basis states.}
  \label{fig:grover}
\end{figure}
A sketch of the ideas for the algorithm is depicted in Figure
\ref{fig:grover}: we have eight basis states, and suppose the fourth
basis state is the target basis state $\ket{\v{\ell}}$. The
representation is purely meant to convey intuition, and does not
geometrically represent the vectors encoding the quantum state, but
solely the amplitude of the coefficients. In Figure \ref{fig:gr1}, all
basis states have the same coefficient. In Figure \ref{fig:gr2}, the
coefficient of the target basis state has its sign flipped. In Figure
\ref{fig:gr3}, we can see that the average value of the coefficients
is slightly below the coefficient for the undesired states. Taking
twice the average and subtracting each coefficient now yields the red
bars in Figure \ref{fig:gr4}, where the target basis state
$\ket{\v{\ell}}$ has a coefficient with much larger value than the
rest, and will therefore be measured with higher probability. Of
course, we need to show that these steps can be implemented with
unitary matrices that can be constructed with a polynomial number of
basic gates.

We now describe the steps above in more detail. 

\subsubsection{Initialization}
The algorithm is initialized by applying the operation $H^{\otimes
  (n+1)}(I^{\otimes n}\otimes X)$ onto the state
$\ket{\v{0}}_{n+1}$. We can express the quantum state as follows:
\begin{align*}
  (I^{\otimes n} \otimes X) \ket{\v{0}}_{n+1} &= \ket{\v{0}}_n \otimes \ket{1} \\
  H^{\otimes (n+1)} (I^{\otimes n} \otimes X) \ket{\v{0}}_{n+1} &= 
  \sum_{\vj \in \{0,1\}^n} \frac{1}{\sqrt{2^n}}
  \ket{\vj}_n \otimes \frac{(\ket{0} - \ket{1})}{\sqrt{2}} =
  \sum_{\vj \in \{0,1\}^n} \alpha_{\vj}
  \ket{\vj}_n \otimes \frac{(\ket{0} - \ket{1})}{\sqrt{2}} = \ket{\psi},
\end{align*}
where $\alpha_{\vj} = \frac{1}{\sqrt{2^n}}$. Thus, the initial
coefficients $\alpha_{\vj}$ of the state $\ket{\psi}$ are real
numbers. Since all the other steps of the algorithm will map real
numbers to real numbers, we only need to consider real numbers through
the course of the algorithm.

\subsubsection{Sign flip: step (i)}
To flip the sign of the target state $\ket{\v{\ell}}_n \otimes
\frac{1}{\sqrt{2}}(\ket{0} - \ket{1})$, we apply $U_f$ to
  $\ket{\psi}$. We now show why this works:
\begin{align*}
  U_f \ket{\psi} &= U_f\left(\sum_{\vj \in \{0,1\}^n} 
  \alpha_{\vj} \ket{\vj}_n \otimes \frac{1}{\sqrt{2}}(\ket{0} - \ket{1})\right) \\
  &= \alpha_{\v{\ell}} \ket{\v{\ell}}_n
  \otimes \frac{1}{\sqrt{2}}(\ket{1} - \ket{0}) + \sum_{\substack{\vj \in \{0,1\}^n \\ \vj \neq \v{\ell}}}
  \alpha_{\vj} \ket{\vj}_n \otimes \frac{1}{\sqrt{2}} (\ket{0} - \ket{1}) \\
  &= \left(-\alpha_{\v{\ell}} \ket{\v{\ell}}_n
   + \sum_{\substack{\vj \in \{0,1\}^n \\ \vj \neq \v{\ell}}}
  \alpha_{\vj} \ket{\vj}_n\right) \otimes \frac{1}{\sqrt{2}}(\ket{0} - \ket{1}).
\end{align*}
As the expression above suggests, we can always think of the last
qubit as being in the state $\frac{1}{\sqrt{2}}(\ket{0} - \ket{1})$
and unentangled from the rest of the qubits, with the sign flip
affecting only the first $n$ qubits. Therefore, the state that we
obtain by applying $U_f$ to $\ket{\psi}$ is the same as $\ket{\psi}$
except that the sign of $\ket{\v{\ell}}_n \otimes
\frac{1}{\sqrt{2}}(\ket{0}-\ket{1})$ has been flipped.

\subsubsection{Inversion about the average: step (ii)}
\label{sec:inversionavg}
To perform the inversion about the average, we want to perform the
following operation:
\begin{equation*}
  \sum_{\vj \in \{0,1\}^n} \alpha_{\vj} \ket{\vj}_n \to \sum_{\vj \in
    \{0,1\}^n} \left(2\left(\sum_{\vk \in \{0,1\}^n} \frac{\alpha_{\vk}}{2^n}\right) - \alpha_{\vj}\right) \ket{\vj}_n,
\end{equation*}
where $\sum_{\vk \in \{0,1\}^n} \frac{\alpha_{\vk}}{2^n}$ is the
average, and therefore we are taking twice the average and subtracting
each coefficient from it. It is not clear yet that this is a unitary
operation, but it will become evident in the following. This mapping
is realized by the following matrix:
\begin{equation*}
  W = \begin{pmatrix} \frac{2}{2^n} - 1 & \frac{2}{2^n} & \dots & \frac{2}{2^n} \\
    \frac{2}{2^n} & \frac{2}{2^n} -1 & \dots & \frac{2}{2^n} \\
    \vdots & \vdots & \ddots & \vdots \\
    \frac{2}{2^n} & \frac{2}{2^n} & \dots & \frac{2}{2^n} - 1
    \end{pmatrix} =
  \begin{pmatrix}
    \frac{2}{2^n} & \frac{2}{2^n} & \dots & \frac{2}{2^n} \\
    \frac{2}{2^n} & \frac{2}{2^n} & \dots & \frac{2}{2^n} \\
    \vdots & \vdots  & \ddots & \vdots \\
    \frac{2}{2^n} & \frac{2}{2^n} & \dots & \frac{2}{2^n}
  \end{pmatrix}
  - I^{\otimes n},
\end{equation*}
where the denominator $\frac{1}{2^n}$ computes the average
coefficient, the numerator $2$ of the fraction takes twice the
average, and finally we subtract the identity to subtract each
individual coefficient from twice the average. From the definition of
the Hadamard gate in \eqref{eq:hadamard}, we can see that the entry of
$H^{\otimes n}$ in position $j,k$ is $\left(H^{\otimes n}\right)_{j,k}
= \frac{1}{\sqrt{2^n}}(-1)^{\vj \bullet \vk}$.  But now if we let:
\begin{equation*}
  R = \begin{pmatrix} 2 & 0 & \dots & 0 \\ 0 & 0 & \dots & 0 \\
    \vdots & & \ddots & \vdots \\
    0 & 0 & \dots & 0 \end{pmatrix} \in \R^{2^n \times 2^n},
\end{equation*}
then we can write $(H^{\otimes n}RH^{\otimes n})_{j,k} =
\left(H^{\otimes n}\right)_{j,0}R_{0,0}\left(H^{\otimes
  n}\right)_{0,k} = \frac{2}{2^n}$, because $R_{j,k} = 0$ for $j \neq
0$ or $k \neq 0$. Therefore, using the fact that $H^{\otimes n}H^{\otimes
  n} = I^{\otimes n}$, we have:
\begin{equation}
  \label{eq:Tdecomp}
  \begin{split}
    W &= H^{\otimes n}RH^{\otimes n} - I^{\otimes n} = H^{\otimes n}(R - I^{\otimes n})H^{\otimes n} = -H^{\otimes n}(I^{\otimes n} - R)H^{\otimes n} \\
    &= -H^{\otimes n} \text{diag}(\underbrace{-1,1,\dots,1}_{2^n}) H^{\otimes n} := -H^{\otimes n} D H^{\otimes n}.
  \end{split}
\end{equation}
The expression \eqref{eq:Tdecomp}, besides providing a decomposition for $W$, also shows that $W$ is unitary,
because $H^{\otimes n}$ is unitary (tensor product of unitary
matrices) and $D$ is diagonal with ones on the diagonal. We must find
a way to construct the matrix $D := \text{diag}(-1,1,\dots,1)$. This
will be discussed in the next section. For now, we summarize our
analysis of the inversion about the average by concluding that it can
be performed by applying $W = -H^{\otimes n} D H^{\otimes n}$ to the
$n$ qubits of interest (i.e., all qubits except the auxiliary qubit
that we used for the sign flip of step (i)).

\optbox{
\subsubsection{Constructing the matrix $D$}
\label{sec:matrixd}
We give a sketch of the idea of how to construct $D =
\text{diag}(-1,1,\dots,1)$. Notice that the effect of this quantum
operation is to flip the sign of the coefficient of the basis state
$\ket{\v{0}}_n$, and leave other coefficients untouched. 

Instead of flipping the sign of $\ket{\v{0}}_n$, let us start by
seeing how to flip the sign of $\ket{\v{1}}_n$ while leaving all other
coefficients untouched. Let C${}^{n-1}Z$ be the gate that applies $Z$
to qubit $n$ if qubits $1,\dots,n-1$ are $\ket{1}$, and does nothing
otherwise. This is similar to the CNOT gate, except that it has
multiple controls, and it applies a $Z$ gate rather than an $X$ (i.e.,
NOT) gate when the control qubits are $\ket{1}$. It is called a
``multiply-controlled $Z$''. C${}^{n-1}Z$ in the case of two qubits
($n=2$) is given by the following matrix:
\begin{equation*}
  \text{C}Z = 
  \begin{pmatrix}
    1 & 0 & 0 & 0 \\
    0 & 1 & 0 & 0 \\
    0 & 0 & 1 & 0 \\
    0 & 0 & 0 & -1 \\
  \end{pmatrix}. 
\end{equation*}
Notice that in the two-qubit case ($n = 2$), the two circuits depicted
in Figure \ref{fig:cz} are equivalent: carrying out the matrix
multiplications will confirm that the circuit on the right in Figure
\ref{fig:cz} implements exactly the C$Z$ matrix as defined
above. Thus, the controlled $Z$ gate can be easily realized with a
CNOT and two Hadamard gates. If we have access to the C${}^{n-1}Z$
gate, we can write:
\begin{equation*}
  D = X^{\otimes n} (\text{C}^{n-1}Z) X^{\otimes n},
\end{equation*}
because, as can be easily verified, this operations flips the sign of
the coefficient of a basis state if and only if all qubits have value
$\ket{0}$ in the basis state. In circuit form, it can be written as
depicted in Figure \ref{fig:grovercircuit}. 
}
\begin{figure}[h!]
    \begin{mdframed}[linecolor=gray!20,backgroundcolor=gray!20]
      \leavevmode
      \centering
      \Qcircuit @C=1em @R=.7em {
        & \qw      & \qw      & \ctrl{1}    & \qw & \qw      & \qw \\
        & \qw      & \qw      & \gate{Z}    & \qw & \qw & \qw \\
      }
      \hspace{10em}
      \Qcircuit @C=1em @R=.7em {
        & \qw      & \qw      & \ctrl{1} & \qw & \qw      & \qw \\
        & \gate{H} & \qw      & \targ    & \qw & \gate{H} & \qw \\
      }
      \caption{Controlled Z gate on two qubits: two possible representations.}
      \label{fig:cz}
    \end{mdframed}
\end{figure}
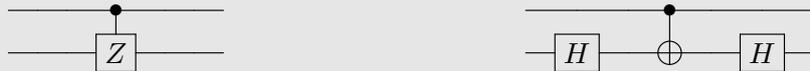
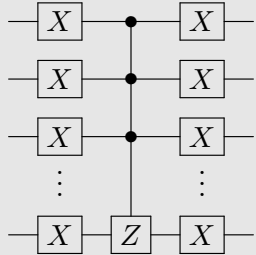
\begin{figure}[h!]
  \begin{mdframed}[linecolor=gray!20,backgroundcolor=gray!20]
    \leavevmode
    \centering
    \Qcircuit @C=1em @R=.7em {
      & \gate{X} & \ctrl{1} & \gate{X} & \qw \\
      & \gate{X} & \ctrl{1} & \gate{X} & \qw \\
      & \gate{X} & \ctrl{3} & \gate{X} & \qw \\
      & \vdots   &          & \vdots   & \\
      &  \\
      & \gate{X} & \gate{Z} & \gate{X} & \qw \\
    }
    \caption{Quantum circuit implementing the $D$ operation used in Grover's algorithm.}
    \label{fig:grovercircuit}
  \end{mdframed}
\end{figure}

\optbox{
Of course, one has to construct the operation
$\text{C}^{n-1}Z$. There are several ways to do so. Perhaps the
simplest construction, suggested in \cite{barenco95elementary}, is to
implement a $\text{C}^{n-2}\text{NOT}$ and a controlled Z gate. The
$\text{C}^{n-2}\text{NOT}$ is actually easy to implement with some
auxiliary qubits, in a construction that will be used in
Section~\ref{sec:code} as well. We show this scheme in Figure
\ref{fig:czdecomposition} with an example for for $n=4$ qubits, but
clearly it can be generalized to an arbitary number of qubits. We
first implement a $\text{C}^{n-2}\text{NOT}$ gate, with an auxiliary
qubit (which is initialized to $\ket{0}$, as one can see from the
bottom qubit in Figure \ref{fig:czdecomposition}) as the target of the
$\text{C}^{n-2}\text{NOT}$. We then implement a $\text{CC}Z$ gate
using a CCNOT and two Hadamard gates on the target qubit; the reader
can easily convince himself or herself that this implements a
doubly-controlled $Z$, using the identity $HXH = Z$ and carrying out
the calculations (in the large unitary matrix for $\text{CC}Z$, the
gate being controlled appears in the bottom right, just as in CNOT).
Summarizing, this yields a decomposition of $\text{C}^{n-1}Z$ with a
linear number of gates and auxiliary qubits. It is possible to forsake
the initialization of the auxiliary qubit, see
\cite{barenco95elementary} for details. To conclude, the construction
of $D$, and therefore of the whole circuit implementing step (ii) of
Grover's search, can be performed in $O(n)$ gates and auxiliary
qubits. We remark that the $-1$ multiplication factor appearing in front
of $H^{\otimes n}$ in 
\eqref{eq:Tdecomp} can be ignored, as it is a global phase factor, see
Example~\ref{ex:globalphase}.}
\begin{figure}[h!]
  \begin{mdframed}[linecolor=gray!20,backgroundcolor=gray!20]
    \leavevmode
    \centering
    \Qcircuit @C=1em @R=1em @!R {
      & \ctrl{4} & \qw      & \qw       & \qw      & \ctrl{4} & \qw \\ 
      & \ctrl{3} & \qw      & \qw       & \qw      & \ctrl{3} & \qw \\
      & \qw      & \qw      & \ctrl{1}  & \qw      & \qw      & \qw \\
      & \qw      & \gate{H} & \targ     & \gate{H} & \qw      & \qw \\
      \lstick{\ket{0}} & \targ    & \qw      & \ctrl{-1} & \qw      & \targ    & \qw & \rstick{\ket{0}}
      \gategroup{3}{3}{5}{5}{2.2em}{--}
    }
    \caption{Decomposition of $\text{C}^{n-1}Z$ for $n=4$. The fifth
      (bottom) qubit is initialized to $\ket{0}$ and is used as
      working space. This implements $\text{C}^3Z$ for the top four
      qubits.}
    \label{fig:czdecomposition}
  \end{mdframed}
\end{figure}
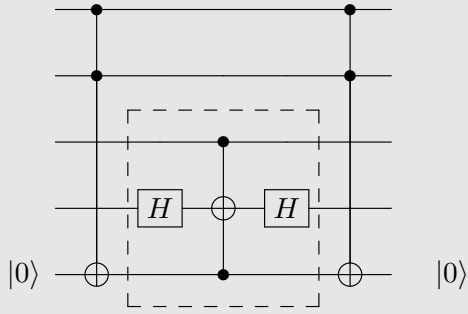

\subsection{Determining the number of iterations}
Let $Q$ be the matrix that applies a single iteration of Grover's
search, consisting of steps (i) and (ii) above. It is paramount to
determine how many iterations should be performed, so that the
coefficient of the desired basis state $\ket{\v{\ell}} \otimes
\frac{1}{\sqrt{2}}(\ket{0} - \ket{1})$ is as large as possible, and
the binary string $\v{\ell}$ is the outcome of a measurement with high
probability. This is what we attempt to do in this section.

Since the last, auxiliary qubit is always in state
$\frac{1}{\sqrt{2}}(\ket{0} - \ket{1})$ and unentangled with the rest,
we can ignore it.  Let
\begin{equation*}
  \ket{\psi_D} = \ket{\v{\ell}}_n, \qquad \ket{\psi_U} =
  \left(\sum_{\substack{\vj \in \{0,1\}^n \\ \vj \neq \v{\ell}}}
  \frac{1}{\sqrt{2^n-1}}\ket{\vj}_n \right)
\end{equation*}
be the desirable and undesirable quantum states, respectively. We
claim that after iteration $k$ of the algorithm, the quantum state can
be expressed as $\ket{\psi_k} = d_k \ket{\psi_D} + u_k
\ket{\psi_U}$. We will show this by induction. Initially, $d_0 =
\frac{1}{\sqrt{2^n}}$ and $u_0 = \sqrt{\frac{2^n - 1}{2^n}}$, where
notice that to obtain $u_0$ from the value of an individual
coefficient in $\ket{\psi_U}$ (all such coefficients are
$\frac{1}{\sqrt{2^n}}$ initially) we have multiplied by $\sqrt{2^n -
  1}$ for normalization. Thus, the claim is true for $k=0$. We now
need to show the induction step: assuming $\ket{\psi_{k-1}} = d_{k-1}
\ket{\psi_D} + u_{k-1} \ket{\psi_U}$, we must show $\ket{\psi_k} = d_k
\ket{\psi_D} + u_k \ket{\psi_U}$.

\optbox{

  The calculations in this part are heavier than usual; if the
  reader is not interested in the details, he or she can simply trust
  the results and skip to the end of this gray box.

  At step (i) of the algorithm, the algorithm flips $d_k \ket{\psi_D}
  + u_k \ket{\psi_U} \to -d_k \ket{\psi_D} + u_k \ket{\psi_U}$.  
  
  At step (ii), the algorithm maps $\alpha_h \to 2A_k - \alpha_h$ for
  each coefficient $\alpha_h$, where $A_k$ is the average
  coefficient. Therefore:
  \begin{align*}
    -\alpha_{\v{\ell}} &\to 2A_k + \alpha_{\v{\ell}} \\
    \alpha_{\v{h}} &\to 2A_k - \alpha_{\v{h}} \quad \forall {\v{h}} \neq \v{\ell}.
  \end{align*}
  To compute $A_k$, we need to determine the value of each individual
  coefficient. The coefficient for $\ket{\v{\ell}}$ is clearly $d_k$,
  as there is only one such state. On the other hand, there are
  $2^n-1$ states with coefficient $u_k$, so the value of the
  coefficient for each of the states $\ket{\vj}, \vj \neq \v{\ell}$ is
  $\frac{u_k}{\sqrt{2^n -1}}$ (the square root is due to
  normalization, see above). The average coefficient at iteration $k$
  is therefore:
  \begin{equation*}
    A_k = \frac{(2^n-1)\frac{1}{\sqrt{2^n - 1}} u_k - d_k}{2^n} = \frac{\sqrt{2^n - 1} u_k - d_k}{2^n}.
  \end{equation*}
  To obtain $u_k$ from $\alpha_{\v{h}}$ we need to multiply by $\sqrt{2^n
    -1}$, so the mapping of step (ii) can be written, overall, as:
  \begin{align*}
    -d_k\ket{\psi_D} + u_k \ket{\psi_U} \to &(2A_k + d_k)\ket{\psi_D} +
    \sqrt{2^n -1}(2A_k - \frac{u_k}{\sqrt{2^n -1}}) \ket{\psi_U} =\\
    &d_{k+1} \ket{\psi_D} + u_{k+1} \ket{\psi_U},
  \end{align*}
  where we defined:
  \begin{align*}
  d_{k+1} &= 2A_k + d_k \\
  u_{k+1} &= 2A_k\sqrt{2^n - 1} - u_k.
  \end{align*}
  This shows the induction step.

  Performing the substitution of $A_k$, we obtain:
  \begin{align*}
    d_{k+1} &= 2\frac{\sqrt{2^n - 1} u_k - d_k}{2^n} + d_k = \left(1 - \frac{1}{2^{n-1}}\right)d_k + \frac{2\sqrt{2^n - 1}}{2^n} u_k \\
    u_{k+1} &= 2\frac{\sqrt{2^n - 1} u_k - d_k}{2^n}\sqrt{2^n - 1} - u_k = -\frac{2\sqrt{2^n - 1}}{2^n} d_k + \left(1-\frac{1}{2^{n-1}}\right)u_k.
  \end{align*}
  This transformation is exactly a clockwise rotation of the vector
  $\begin{pmatrix} d_k \\ u_k \end{pmatrix}$ by a certain angle
  $\theta$, because it has the form $$\begin{pmatrix} \cos \theta & \sin
    \theta \\ -\sin \theta & \cos\theta \end{pmatrix} \begin{pmatrix}
    d_k \\ u_k \end{pmatrix}$$ and it satisfies the relationship $\sin^2
  \theta + \cos^2 \theta = 1$. The angle $\theta$ must satisfy:
  \begin{equation}
    \label{eq:groversintheta}
    \sin \theta = \frac{2\sqrt{2^n - 1}}{2^n}.
  \end{equation}
  Notice that because this value of the sine is very small (for large
  $n$), we can use the approximation $\sin x \approx x$ (when $x$
  is close to 0) to write:
  \begin{equation}
    \label{eq:grovertheta}
    \theta = \frac{2\sqrt{2^n - 1}}{2^n}.
  \end{equation}
}

Overall, the above analysis shows that each iteration performs a
rotation by an angle $\theta$ of the vector $\ket{\psi_k}$, which
always belongs to the plane spanned by $\ket{\psi_D}$ and
$\ket{\psi_U}$. So after $k$ iterations the coefficients $d_k, u_k$
satisfy the following equation:
\begin{equation*}
  \begin{pmatrix} d_k \\ u_k \end{pmatrix} = \begin{pmatrix} \cos \theta & \sin \theta \\ - \sin \theta & \cos \theta \end{pmatrix}^k \begin{pmatrix} d_0 \\ u_0 \end{pmatrix},
\end{equation*}
which can be rewritten as:
\begin{align*}
  d_{k} &= \cos k \theta d_0 + \sin k \theta u_0 \\
  u_{k} &= -\sin k \theta d_0 + \cos k \theta u_0.
\end{align*}
In order to maximize the probability of obtaining $\ket{\psi_D}$ after
a measurement, remember that $|u_0| \gg |d_0|$, so the best choice is
to pick $k \theta = \frac{\pi}{2}$ which yields the largest value of
$|d_k|$. Using \eqref{eq:grovertheta}, the optimal number of
iterations of Grover's search algorithm is:
\begin{equation}
  \label{eq:groveroptiter}
  k \approx \frac{2^{n} \pi}{4\sqrt{2^n - 1}} \approx
  \frac{\pi}{4}\sqrt{2^n}.
\end{equation}
After this many iterations, we have a probability close to 1 of
measuring $\ket{\psi_D}$ and obtaining the sought state
$\ket{\v{\ell}}$. Comparing this with a classical algorithm, that requires
$O(2^n)$ iterations, we obtained a quadratic speedup.
\begin{remark}
  If we perform more iterations of Grover's algorithm than the optimal
  number, the probability of measuring the desired state actually goes
  down, and reduces our chances of success. Therefore, it is important
  to choose the right number of iterations.
\end{remark}
Of course, the approximation for $\theta$ given in
\eqref{eq:grovertheta} is only valid for large $n$: for smaller $n$,
it is better to compute the optimal number of iterations deriving
$\theta$ from \eqref{eq:groversintheta}. We conclude this section by
noticing that in case there are multiple input values on which $f$ has
value 1, we should amend the above analysis adjusting the values for
$d_0$ and $u_0$, but the main steps remain the same.

\section{Numerical implementation of Grover's algorithm}
\label{sec:code}
We conclude the tutorial by showing how to implement Grover's
algorithm using the open-source Python library Qiskit
\cite{qiskit}. One of the advantages of using Qiskit is that the code
is self-explanatory. The library requires Python 3 and can be
installed via {\tt pip}:
\begin{lstlisting}
  pip install qiskit==0.11.1
\end{lstlisting}
The code below was tested with Qiskit 0.11.1; it may work with other
versions as well.

We apply Grover's algorithm to the problem of finding a satisfying
assignment for an instance of the Exactly-1 3-SAT problem, which is
defined as follows:
\vskip1ex
\fbox{ 
  \begin{minipage}{0.9\linewidth}
    {\bf Problem (Exactly-1 3-SAT): Determining a satisfying assignment
      containing one true literal per clause.} \\ Input: SAT formula
    in conjunctive normal form $\bigwedge_{k=1}^m C_k$ over $n$ Boolean
    variables $x_1,\dots,x_n$, with 3 literals per clause
    $C_1,\dots,C_m$. \\ Output: Does there exist an assignment of
    $x_1,\dots,x_n$ such that every clause $C_1,\dots,C_m$ has exactly
    one {\tt True} literal?
  \end{minipage}
}

This problem is NP-hard \cite{garey1972computers}. In our
implementation, an instance of Exactly-1 3-SAT is specified as a list
of clauses, where each clause contains three integers: a positive
integer is the index of a positive literal, a negative integer is the
index of a negative literal. For example, the Python list of lists
\begin{equation*}
  [[1, 2, -3], [-1, -2, -3], [-1, 2, 3]]
\end{equation*}
represents the instance:
\begin{equation}
  \label{eq:satinstance}
  (x_1 \bigtriangledown x_2 \bigtriangledown \neg x_3) \land (\neg x_1
  \bigtriangledown \neg x_2 \bigtriangledown \neg x_3) \land (\neg x_1 \bigtriangledown x_2 \bigtriangledown x_3).
\end{equation}
We use this formula in the example given below. We use the symbol
$\bigtriangledown$ rather than the usual $\vee$ (normally used to
indicate the logical OR) to emphasize that this is not a regular 3-SAT
formula, but an Exactly-1 3-SAT formula: the problem definition
requires {\em exactly} one True literal per clause.

In the implementation presented in this section we only allow up to
three boolean variables and three clauses, i.e., $n \le 3, m \le
3$. At the end of the section we will discuss how to generalize the
implementation to allow an arbitrary number of clauses, which is left
as an exercise. Notice that the suggested approach requires additional
qubits for each clause. The code presented here for $n \le 3, m \le 3$
yields a circuit with at most 8 qubits, that can be simulated on a
laptop in 1-2 minutes in most cases. Further clauses would require
additional qubits and slow down the simulation of the circuit (roughly
by a factor of 2 for each additional qubit).

To apply Grover's algorithm we will use three basic subroutines: a
subroutine to construct the initial state, a subroutine to compute the
unitary $U_f$ implementing the black-box function $f$, and a
subroutine to perform the inversion about the average. We will discuss
them in order.

\subsection{Initial state}
Before we construct the initial state, let us give names to some of
the quantum registers (i.e., groups of quantum qubits) that we
need. Grover's algorithm applies to a function with an $n$-qubit input
and a single-qubit output. We call {\tt f\_in} the input register of
$U_f$, of size $n$, and {\tt f\_out} the output register of $U_f$, of
size 1.  We construct the initial state as follows.
\begin{lstlisting}[language=Python]
def input_state(circuit, f_in, f_out, n):
    """(n+1)-qubit input state for Grover search."""
    for j in range(n):
        circuit.h(f_in[j])
    circuit.x(f_out)
    circuit.h(f_out)
\end{lstlisting}
This is equivalent to the circuit given in Figure \ref{fig:groverinit}.
\begin{figure}[h!tb]
\leavevmode
\centering
\Qcircuit @C=1em @R=0.7em {
\lstick{{\tt f\_in} \quad \ket{0}_n} & \qw & {/^n} \qw & \gate{H^{\otimes n}} & \qw & {/^n} \qw & \qw \\
\lstick{{\tt f\_out} \quad \ket{0}_1} & \qw & \gate{X}  & \gate{H}             & \qw & \qw       & \qw \\
}
\caption{Initialization step of Grover's algorithm.}
\label{fig:groverinit}
\end{figure}
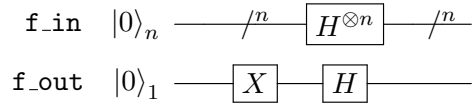

\subsection{Black-box function $U_f$}
Implementing $U_f$ for the Exactly-1 3-SAT problem is the most complex
part of the code, and it can be done in several ways. To favor
simplicity, we decompose the problem of computing $U_f$ by introducing
$m$ auxiliary qubits (these are often called ``ancillas'' in the
quantum computing literature), one for each clause. For each clause we
construct a circuit that bit-flips the corresponding auxiliary qubit
if and only if the clause has exactly one true literal (these
auxiliary qubits will be initialized in state $\ket{0}$). Finally, we
apply a bit-flip on the output register of $U_f$ if and only if all
$m$ auxiliary qubits are $1$.  In Figure \ref{fig:groverufblock} we
show a circuit that bit-flips the bottom qubit $y$ if the clause $x_1
\bigtriangledown \neg x_2 \bigtriangledown x_3$ is satisfied.
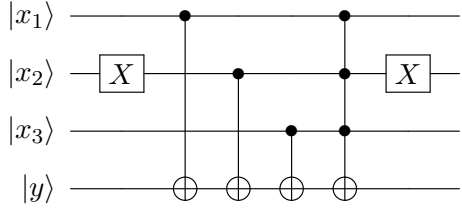
\begin{figure}[h!tb]
\leavevmode
\centering
\Qcircuit @C=1em @R=.7em @!R {
\lstick{\ket{x_1}} & \qw      & \ctrl{3} & \qw      & \qw      & \ctrl{1} & \qw      & \qw \\
\lstick{\ket{x_2}} & \gate{X} & \qw      & \ctrl{2} & \qw      & \ctrl{1} & \gate{X} & \qw \\
\lstick{\ket{x_3}} & \qw      & \qw      & \qw      & \ctrl{1} & \ctrl{1} & \qw      & \qw \\
\lstick{\ket{y}} &   \qw    & \targ      & \targ    & \targ    & \targ    & \qw      & \qw \\
}
\caption{Quantum circuit bit-flipping the bottom qubit if the clause
  $x_1 \bigtriangledown \neg x_2 \bigtriangledown x_3$ is satisfied by exactly one literal.}
\label{fig:groverufblock}
\end{figure}
The idea is as follows. The $X$ gate bit-flips the qubit $\ket{x_2}$,
since $x_2$ appears negated in the clause. Using three CNOT gates, we
set $\ket{y} = \ket{y \oplus x_1 \oplus \neg x_2 \oplus x_3}$,
implying that $y$ is bit-flipped if an odd number of literals is
satisfied. We use a triply-controlled NOT gate to finally obtain
$\ket{y} = \ket{y \oplus x_1 \oplus \neg x_2 \oplus x_3 \oplus (x_1
  \land \neg x_2 \land x_3)}$, as desired. To set $\ket{y} = \ket{1}$
if and only if exactly one literal is satisfied, it is enough to
assume that $\ket{y}$ starts in $\ket{0}$ state. The final $X$ gate
resets the state of qubit $\ket{x_2}$.

To implement this circuit in Qiskit, there is a small obstacle: the
triply-controlled NOT gate is not part of the basic gate set. We can
implement it with a strategy similar to what is discussed in Section
\ref{sec:matrixd} for the $\text{C}^{n-1}Z$ gate: we use
three CCNOT gates and one auxiliary qubit, implementing the circuit in
Figure \ref{fig:cccnot}. This has the drawback of requiring an extra
qubit, but it is easy to understand. We remark that CCNOT, while not
part of the basic gate set of Qiskit, is defined as a macro and
therefore can be used as if it were part of the basic gate set: the
software will automatically perform the substitution, using the
circuit of Figure~\ref{fig:ccnotbasic} to implement CCNOT.

We can quickly verify that the circuit in Figure~\ref{fig:cccnot}
bit-flips $\ket{q_5}$ if and only if $\ket{q_1}, \ket{q_2}, \ket{q_3}$
are 1; the final CCNOT resets the state of the auxiliary qubit
$\ket{q_4}$, which is assumed to be initialized at $\ket{0}$ and is
left in state $\ket{0}$.
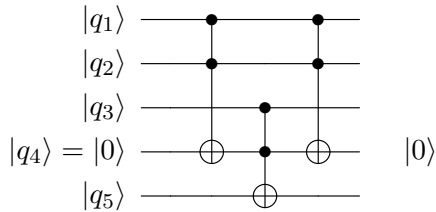
\begin{figure}[h!tb]
\leavevmode
\centering
\Qcircuit @C=1em @R=.7em @!R {
\lstick{\ket{q_1}} & \qw      & \ctrl{1} & \qw      & \ctrl{1} & \qw \\
\lstick{\ket{q_2}} & \qw      & \ctrl{2} & \qw      & \ctrl{2} & \qw \\
\lstick{\ket{q_3}} & \qw      & \qw      & \ctrl{1} & \qw      & \qw \\
\lstick{\ket{q_4} = \ket{0}} & \qw      & \targ    & \ctrl{1} & \targ    & \qw & \rstick{\ket{0}} \\
\lstick{\ket{q_5}} & \qw      & \qw      & \targ    & \qw      & \qw \\
}
\caption{A possible implementation of a triply-controlled NOT gate.}
\label{fig:cccnot}
\end{figure}

The implementation of $U_f$ then proceeds as follows. We loop over the
clauses, using index $k=0,\dots,m-1$ (as it often happens, in this
tutorial we use ``mathematical language'' and the clauses are numbered
$k=1,\dots,m$, but in the Python implementation the corresponding
array is indexed $k=0,\dots,m-1$). For each clause we implement the
circuit in Figure \ref{fig:groverufblock}, setting the auxiliary qubit
{\tt aux[k]} to 1 if clause $C_{k+1}$ is satisfied. We then perform a
multiply-controlled NOT operation to ensure that the output register
{\tt f\_out} is bit-flipped if all $m$ auxiliary qubits are
$1$. Finally, we run the same circuit in reverse to reset the state of
the auxiliary qubits.

\begin{lstlisting}[language=Python]
def black_box_u_f(circuit, f_in, f_out, aux, n, exactly_1_3_sat_formula):
    """Circuit that computes the black-box function from f_in to f_out.

    Create a circuit that verifies whether a given exactly-1 3-SAT
    formula is satisfied by the input. The exactly-1 version
    requires exactly one literal out of every clause to be satisfied.
    """
    num_clauses = len(exactly_1_3_sat_formula)
    if (num_clauses > 3):
        raise ValueError('We only allow at most 3 clauses')
    for (k, clause) in enumerate(exactly_1_3_sat_formula):
        # This loop ensures aux[k] is 1 if an odd number of literals
        # are true
        for literal in clause:
            if literal > 0:
                circuit.cx(f_in[literal-1], aux[k])
            else:
                circuit.x(f_in[-literal-1])
                circuit.cx(f_in[-literal-1], aux[k])
        # Flip aux[k] if all literals are true, using auxiliary qubit
        # (ancilla) aux[num_clauses]
        circuit.ccx(f_in[0], f_in[1], aux[num_clauses])
        circuit.ccx(f_in[2], aux[num_clauses], aux[k])
        # Flip back to reverse state of negative literals and ancilla
        circuit.ccx(f_in[0], f_in[1], aux[num_clauses])
        for literal in clause:
            if literal < 0:
                circuit.x(f_in[-literal-1])
    # The formula is satisfied if and only if all auxiliary qubits
    # except aux[num_clauses] are 1
    if (num_clauses == 1):
        circuit.cx(aux[0], f_out[0])
    elif (num_clauses == 2):
        circuit.ccx(aux[0], aux[1], f_out[0])
    elif (num_clauses == 3):
        circuit.ccx(aux[0], aux[1], aux[num_clauses])
        circuit.ccx(aux[2], aux[num_clauses], f_out[0])
        circuit.ccx(aux[0], aux[1], aux[num_clauses])
    # Flip back any auxiliary qubits to make sure state is consistent
    # for future executions of this routine; same loop as above.
    for (k, clause) in enumerate(exactly_1_3_sat_formula):
        for literal in clause:
            if literal > 0:
                circuit.cx(f_in[literal-1], aux[k])
            else:
                circuit.x(f_in[-literal-1])
                circuit.cx(f_in[-literal-1], aux[k])
        circuit.ccx(f_in[0], f_in[1], aux[num_clauses])
        circuit.ccx(f_in[2], aux[num_clauses], aux[k])
        circuit.ccx(f_in[0], f_in[1], aux[num_clauses])
        for literal in clause:
            if literal < 0:
                circuit.x(f_in[-literal-1])    
\end{lstlisting}

\subsection{Inversion about the average}
The inversion about the average is discussed in Sections
\ref{sec:inversionavg}-\ref{sec:matrixd}. This can be implemented as
follows.

\begin{lstlisting}[language=Python]
def inversion_about_average(circuit, f_in, n):
    """Apply inversion about the average step of Grover's algorithm."""
    # Hadamards everywhere
    for j in range(n):
        circuit.h(f_in[j])
    # D matrix: flips the sign of the state |00...00> only
    for j in range(n):
        circuit.x(f_in[j])
    n_controlled_Z(circuit, [f_in[j] for j in range(n-1)], f_in[n-1])
    for j in range(n):
        circuit.x(f_in[j])
    # Hadamards everywhere again
    for j in range(n):
        circuit.h(f_in[j])
\end{lstlisting}

\vskip1ex

The above routine requires a $\text{C}^{n-1}Z$ gate; we implement it,
for three qubits, using a CCNOT gate and two Hadamards. We raise an
exception if there are more than two controls, which are not necessary
for this example.
\begin{lstlisting}[language=Python]
def n_controlled_Z(circuit, controls, target):
    """Implement a Z gate with multiple controls"""
    if (len(controls) > 2):
        raise ValueError('The controlled Z with more than 2 ' +
                         'controls is not implemented')
    elif (len(controls) == 1):
        circuit.h(target)
        circuit.cx(controls[0], target)
        circuit.h(target)
    elif (len(controls) == 2):
        circuit.h(target)
        circuit.ccx(controls[0], controls[1], target)
        circuit.h(target)  
\end{lstlisting}

\subsection{Putting everything together}
To run Grover's algorithm, we have to initialize a quantum circuit
with the desired number of qubits, and apply the routines described
above. Notice that for three qubits the optimal number of Grover
iterations is two, see \eqref{eq:groveroptiter} (the $\sin \theta$
approximation is not accurate for $n=3$, but doing the calculations
more carefully, we can verify that two iterations is still
optimal). Including the small, necessary setup to initialize a quantum
circuit with Qiskit (``classical registers'' are used to store the
result of a measurement), this results in the following code:

\begin{lstlisting}
import sys
from qiskit import QuantumRegister, ClassicalRegister, QuantumCircuit
from qiskit import compiler, Aer
from qiskit.tools import visualization

# Make a quantum program for the n-bit Grover search.
n = 3
# Exactly-1 3-SAT formula to be satisfied, in conjunctive
# normal form. We represent literals with integers, positive or
# negative to indicate a boolean variable or its negation.
exactly_1_3_sat_formula = [[1, 2, -3], [-1, -2, -3], [-1, 2, 3]]

# Define three quantum registers: 'f_in' is the search space (input
# to the function f), 'f_out' is bit used for the output of function
# f, aux are the auxiliary bits used by f to perform its
# computation.
f_in = QuantumRegister(n)
f_out = QuantumRegister(1)
aux = QuantumRegister(len(exactly_1_3_sat_formula) + 1)
# One classical register to store the result of a measurement
ans = ClassicalRegister(n)
# Create quantum circuit with the quantum and classical registers
# defined above
qc = QuantumCircuit(f_in, f_out, aux, ans, name='grover')

input_state(qc, f_in, f_out, n)
# Apply two full iterations
black_box_u_f(qc, f_in, f_out, aux, n, exactly_1_3_sat_formula)
inversion_about_average(qc, f_in, n)
black_box_u_f(qc, f_in, f_out, aux, n, exactly_1_3_sat_formula)
inversion_about_average(qc, f_in, n)
# Measure the output register in the computational basis
for j in range(n):
    qc.measure(f_in[j], ans[j])

# Create an instance of the local quantum simulator    
quantum_simulator = Aer.get_backend('qasm_simulator')
# Compile the circuit into "quantum object code" that can be
# executed on the simulator
qobj = compiler.assemble(qc, quantum_simulator, shots=2048)
# Execute and store the results. Note that this could take some
# time (up to a few minutes, depending on the machine)
job = quantum_simulator.run(qobj)
result = job.result()
# Get counts
counts = result.get_counts('grover')
print('Observed measurement outcomes:')
print('string | count')
for key in sorted(counts):
    print(' {:>5s}   {:d}'.format(key, counts[key]))

# Plot histogram        
figure = visualization.plot_histogram(counts)
print()
# We can display the histogram with figure.show() if matplotlib is
# properly configured. Instead, we write it to file.
figure.savefig('groverhist.png')
print('Histogram saved as groverhist.png')
\end{lstlisting}

\vskip1ex That's it! We have successfully executed Grover's
algorithm. The resulting histogram is given in Figure
\ref{fig:groverhist}: with only two calls to the black-box function
$U_f$, we sample the correct string 101 (the only solution of
\eqref{eq:satinstance}) with probability $\approx 95\%$. The argument
      {\tt shots} given to the {\tt assemble()} function determines the
      number of samples extracted from the quantum state, i.e., the
      quantum circuit is executed that many times, each time
      performing a measurement (and therefore potentially obtaining
      different outcomes).

\begin{figure}[tb]
  \centering
  \includegraphics[width=0.7\textwidth]{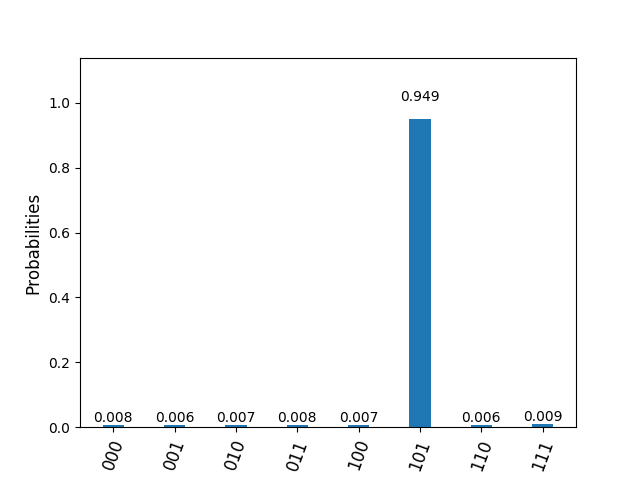}
  \caption{Histogram of output probabilities.}
  \label{fig:groverhist}
\end{figure}

The Qiskit allows running experiments on real quantum computing
hardware accessible on the cloud via the IBM Q experience, changing
the backend object used to run the experiment. In this example, we
used a classical simulation of the quantum computer executed locally
(via the {\tt qasm\_simulator} backend); such a simulation is only
able to scale up to a dozen qubits or so (a regular laptop should be
able to simulate $\approx 18$ qubits, but the computation can be 
slow after 12-14 qubits).

To run experiments on a real device, one first needs an account on the
IBM Q experience, after which {\tt IBMQ} can be loaded from Qiskit and
one can access the list of available remote backends --- which includes
some of IBM's devices. More detailed instructions are available on
Qiskit's webpage. We should also remark that the current generation of
quantum computing devices has limited qubit connectivity, i.e., it
only allows the application of two-qubit gates between certain pairs
of qubits. The code listed in this tutorial assumes an all-to-all
connectivity, which makes for much simpler code. The interested reader
can modify the code to allow for limited connectivity, using SWAP
gates when necessary.

Another extension of the code that we leave as an exercise (and is
perhaps simpler than the above exercise) is to allow an arbitrary
number of boolean variables and clauses in the Exactly-1 3-SAT
formula. For this extension, we need a controlled NOT with an
arbitrary number of controls. Such a gate will be used in two places:
first, to implement the function $U_f$ (in our approach, we use one
auxiliary qubit per clause and a final controlled NOT operation to
check that all clauses are $\ket{1}$); second, to implement the
$n$-controlled Z. One way to implement a controlled NOT with an
arbitrary number of controls is to apply the same idea as the circuit
in Figure~\ref{fig:cccnot} for the triply-controlled NOT. This will
require extra qubits, which will slow down the classical simulation of
the resulting circuit.

\section{Further reading}
\label{sec:furtherreading}
In this tutorial we use several notational devices to help the reader,
but they are usually not employed in the quantum computing
literature. We list them here.
\begin{itemize}
\item The subscript for bra-ket vectors to indicate the dimension of
  the space, e.g., $\ket{\psi}_q$ for $2^q$-dimensional
  vectors. Typically, the dimension of the space is defined elsewhere
  and can be understood from the context. Whenever subscripts for kets
  are used, it is normally to address registers.
\item The vector arrow, e.g., $\vj$, to indicate binary
  strings. Typically, binary strings are not distinguished from other
  mathematical symbols and they can be identified from the context.
\item The use of Roman letters for basis vectors and Greek letters for
  general, i.e., possibly not basis, vectors.
\item The notation for the probability of observing the measurement
  outcomes: this tutorial makes explicit the qubit(s) that is (are)
  being measured, but typically this is only defined by the context.
\end{itemize}
Finally, the all-zero binary string of dimension $q$ is normally
denoted $0^q$, rather than $\v{0}$.

We end this tutorial with some pointers to papers describing quantum
algorithm that provide a speedup with respect to classical
computation, in various areas of computing. This list is by no means
exhaustive: it merely serves the purpose of giving the reader some
idea about existing work showcasing the power of quantum
computation. Additional references can be found from the papers listed
below, or looking up on the arXiv, where most of the quantum
literature is available.

The first article that we mention is of course Shor's paper on integer
factorization \cite{shor97polynomial}. Shor's algorithm determines the
prime factors of an integer using a combination of a classical
algorithm and a quantum subroutine for the following period-finding
problem: given integers $a, n$ with $a < n$, find the smallest integer
$r$ such that $a^x \mod n = a^{x+r} \mod n$, for all $x$. We call this
``period finding'' because it can be seen as determining the period of
the function $f(x) = a^x \mod n$.

For readers who are familiar with topology, the work by Aharonov,
Jones and Landau on the approximation of Jones polynomial
\cite{aharonov2009polynomial} will be of great interest. The Jones
polynomial is an invariant of an oriented knot. It is known that
exactly evaluating the Jones polynomial is very difficult: it is a
\#P-hard problem, and it is not expected to admit a polynomial-time
classical algorithm. The paper \cite{aharonov2009polynomial} shows how
to approximately solve the problem in polynomial time using a quantum
computer --- an easier task because of the approximation, but one for
which no classical algorithm is known. Furthermore, computing such an
approximation is a BQP-complete problem, i.e., it solves a problem
that is as hard as any other problem that admits an efficient quantum
algorithm. It is believed that BQP-complete problems cannot be solved
efficiently on classical computers, because that would imply BQP = P,
which in turn implies that there exist a polynomial-time classical
algorithm for integer factorization. Thus, surprisingly the
computational power of quantum computing can be rephrased in terms of
the ability to approximate this topological quantity.

Another BQP-complete problem that received a lot of attention is
Hamiltonian simulation. This is the problem of simulating the time
evolution of a quantum system, for which quantum computers where
originally proposed \cite{feynman1982simulating}: it is known since
the early days that quantum computer can efficiently solve it
\cite{lloyd1996universal}, whereas all known classical algorithms
require exponential time. There are many papers on this topic; we
refer to the very recent \cite{haah2018quantum} for an entry point
with a good list of references.

Exponential speedups can also be obtained for easier-to-describe
problems. As an example, we mention the work of Harrow, Hassidim and
Loyd on the solution of linear systems \cite{harrow2009quantum} in
logarithmic time (although attaining logarithmic time requires several
assumptions on the input and the output of the algorithm, because
parsing the equations in the linear system trivially requires linear
time already), and several papers discussing quantum acceleration for
random walks on graph \cite{childs2003exponential}, which find many
applications in computer science, e.g., \cite{ambainis2007quantum}.

\subsection*{Acknowledgments}
We are extremely grateful to two anonymous referees, whose patience
and numerous detailed remarks on an earlier draft of this manuscript
helped significantly improve its quality, and to Sergey Bravyi for
many illuminating discussions.

\bibliographystyle{apalike}
\bibliography{quantum}

\end{document}